\documentclass{lmcs}
\pdfoutput=1
\usepackage[utf8]{inputenc}

\usepackage{lastpage}


\usepackage{hyperref}
\usepackage{xspace,amsmath,amssymb}

\theoremstyle{plain}

\newcommand{\eg}{\text{e.g.,}\xspace}

\begin{document}

\title[A correspondence between the time and space complexity]{A correspondence between the time and space \texorpdfstring{\\} {}complexity\rsuper*} 

\thanks{The author thanks Mikhail G. Peretyat'kin and Alexander V. Seliverstov for their attention to the work and useful discussion.} 

\author[I.V.~Latkin]{Ivan V. Latkin\lmcsorcid{0000-0001-8844-4105}}

\address{East Kazakhstan Technical University, Ust-Kamenogorsk, Kazakhstan}
\email{lativan@yandex.kz}  

\begin{abstract}
\noindent We investigate the correspondence between the time and space recognition 
complexity of languages. For this purpose, we will code the long-continued computations of 
deterministic two-tape Turing machines by the relatively short-length quantified Boolean 
formulae. The modified Meyer and Stockmeyer method will appreciably be used for this 
simulation. It will be proved using this modeling that the complexity classes 
Deterministic Exponential Time and Deterministic Polynomial Space coincide. It will also  
be proven that any language recognized in polynomial time can be recognized in almost 
logarithmic space. Furthermore, this allows us slightly to improve the early founded lower 
complexity bound of decidable theories that are nontrivial relative to some equivalence relation 
(this relation may be equality) --- each of these theories is consistent with the formula, which 
asserts that there are two non-equivalent elements.

This is a corrected version of the preprint \url{https://arxiv.org/abs/2311.01184} submitted 
originally on November 2, 2023. 
\end{abstract}

\maketitle

\section{Introduction}\label{s1}

At the beginning, we recall some designations. Although they are widely applied, 
this reminder will help to avoid misunderstandings. 

A function $\exp_k(n)$ is called {\em $k$-iterated} (or {\em $k$-story}, or {\it $k$-fold}) 
exponential if, for every natural $k$, it is calculated in the following way: \ 
$\exp_0(n)\!=\!n$, \ $\exp_{k+1}(n)\!=\!2^{\exp_{k}(n)}$. The length of a word $X$ is denoted by 
$|X|$, i.e., $|X|$ is the number of symbols in $X$ together indices. If $A$ is a set, then   
$|A|$ denotes its cardinality; "$A \rightleftharpoons\mathcal{A}$" means "$A$ is a designation 
for $\mathcal{A}$" or "$A$ equals to $\mathcal{A}$ by definition"; and $\exp(n)\rightleftharpoons
\exp_1(n)$. If $a$ is a real number, then $\lceil a\rceil$ will denote the least integer  
that is no less than $a$. We will not use logarithms on different bases, so $\log x$ further 
denotes $\log_2x$.

As usual, if $a$ is a symbol (or letter) of some alphabet, then $a^k$ denotes the $k$-fold 
concatenation of this letter with themself; $|a^k|\!=\!k$, $a^0$ is the empty word $\Lambda$.
The variable $n$ always designates the length of the input 
string in the expressions $r$-$DTIME[t(n)]$ and $r$-$DSPACE[s(n)]$, which accordingly are 
the classes of languages recognized by the $r$-tape deterministic Turing machines ($r$-DTM) 
with the time and space complexity $\mathcal{O}(t(n))$ and $\mathcal{O}(s(n))$. We regard that 
\ $r$-$DTIME[t(n)]\subseteq$ $r$-$DTIME[t_1(n)]$ \ and \ $r$-$DSPACE[s(n)]\subseteq$ 
$r$-$DSPACE[s_1(n)]$, if $t(n)\!<\!t_1(n)$ and $s(n)\!<\!s_1(n)$ for almost all $n$. When the 
number of tapes is inessential, we will omit the label "$r$-". 

A set (or a language) is identified with the related decision (recognition) problem. 
We recall the definitions and notations of some complexity classes: \  $\mathbf{L}
\!\rightleftharpoons\!DSPACE[\log n]$ (Deterministic Logarithmic Space); 
\[
\mathbf{^kEXP}\rightleftharpoons
\bigcup_{j\geqslant1}\bigcup_{c\geqslant1} DTIME[\exp_k(cn^j)]
\] 
($k$-fold Deterministic Exponential Time) for $k\!\geqslant\!1$;  
\[
\mathbf{^kEXPSPACE}\rightleftharpoons\bigcup_{j\geqslant1}\bigcup_{c\geqslant1}
DSPACE[\exp_k(cn^j)]
\] 
($k$-fold Deterministic Exponential Space) for $k\!\geqslant\!1$. One can define also that 
$\mathbf{^0EXP}$ is $\mathbf{P}$, \ $\mathbf{^0EXPSPACE}=\mathbf{PSPACE}$.

The language $TQBF$ consists of the true quantified Boolean formulae (the corresponding 
problem is designated as $QBF$, or sometimes $QSAT$).

\subsection{Is the equality of the $\mathbf{EXP}$ and $\mathbf{PSPACE}$ classes possible?}
\label{s1.1}

The coincidence of these classes seems to some researchers quite fantastic. 
However, certain reasonable arguments can shake this disbelief. 

Let us consider the assertion (which is almost apparent) that the class $\mathbf{P}$ is the 
proper subclass of $\mathbf{PSPACE}$. Because the language $TQBF$ is complete for the class  
$\mathbf{PSPACE}$ \cite{GJ79,SM73}, this statement is tantamount to proving that $TQBF\!\notin\!
\mathbf{P}$. The most reliable way to establish that the language $TQBF$ does not belong to 
$\mathbf{P}$ is to prove that its time complexity is exponential or at least subexponential.
Let us discuss this in more detail. 

Firstly, the hypothesis that the language $TQBF$ does not belong to $DTIME[\exp(cn^{\alpha})]$ 
for some $c,\alpha\!>\!0$, is not less plausible than the Exponential Time Hypothesis which  
asserts that the problem on the satisfiability of 3-$CNF$ with $m$ variables cannot be solved in  
time $\exp(Cm^{\beta})$ for appropriate $C,\beta\!>\!0$. 

Secondly, this has already been made by one of the most natural methods. Namely, it has been 
proved in \cite{Lat22} that the theory of two-element Boolean algebra, which is equivalent to the 
language $TQBF$ relatively polynomial reduction, does not belong to $DTIME[\exp(c\cdot n^{\rho})]
$, where $\rho\!=\!(2\!+\!\varepsilon)^{-1}$ for every $\varepsilon\!>\!0$, and so 
$\mathbf{P}$ is a proper subclass of $\mathbf{PSPACE}$. 
   
The lower bound on the time complexity for the theory $TQBF$ has been yielded by the technique 
of the immediate codings of the machine actions employing formulae similarly as it has been 
made for many other theories in \cite{FR74,Mey73,Mey75,MS72,Rab99,Rob74,Sto74,Vor04}; 
this approach was named the Rabin and Fischer method in \cite{Lat22}\footnote{although one 
should call the G\"{o}del, Rabin, and Fischer method}. For many of these theories, more 
precise lower bounds on the complexity of decidability have now been obtained; 
however, we here wish to emphasize the long-standing prominence of the method.

More precisely, the following theorem has been proved in \cite{Lat22}; here and in what follows, a 
Turing machine is identified with its program (or {\it transitions function}). 

\begin{thm}[\cite{Lat22}]{}\label{old-main}
For each one-tape deterministic Turing machine $P$ and every input string $X$, one can 
write a closed formula (sentence) $\Omega(X,P)$ of the signature of the 
two-element Boolean algebra $\mathcal{B}$ with the following properties:\\
(i) \ $\Omega(X,P)$ can be written within polynomial time of $|X|$ and $|P|$;\\
(ii)\ $\Omega(X,P)$ is true over the algebra $\mathcal{B}$ \ if and only 
if the Turing machine $P$ accepts input $X$ within time \ $\exp(|X|)$;\\
(iii)\ for every $\varepsilon\!>\!0$, there is a constant \ $D\!>\!0$ such that 
the inequalities \
$|X|\!<\!|\Omega(X,P)|\!\leqslant\! D\cdot|P|\cdot|X|^{2+\varepsilon}$
 \ hold for all sufficiently long $X$. 
 \end{thm}

The evaluation of the time complexity of $Th\mathcal{B}$ has been obtained from this  
theorem by the Rabin and Fischer method (see Proposition \ref{mainPr} here and Proposition 
2.1 in \cite{Lat22}). 

We will often cite the article \cite{Lat22} in Sections \ref{s1} and \ref{s2}; nevertheless, the 
given paper is fully self-contained, and no previous knowledge of that work is needed. Moreover, 
the paper \cite{Lat22} is not even mentioned in Sections 3--6 during the proof of the main 
theorem on modeling, although the method of its proof differs from the one of Theorem 
\ref{old-main} only in detail.   

Analyzing the theorem's proof leads to a natural generalization, which we will consider 
in detail in Section \ref{s2}. Now, let us restrict ourselves to the simplest amplification, 
which we will describe here briefly. Namely, the time restriction (it will be named a 
{\it simulation period}) in the second item of Theorem \ref{old-main} can be extended to 
$\exp(m(|X|)$, where  $m(n)$ is a polynomial. 

In this case, instead of the formula $\Omega(X,P)$, we construct the sentence $\Omega^{(m)}(X,P)$
with almost the same properties: the item (i) coincides with the former verbatim, and 
in the items (ii) and (iii), the number $|X|\!=\!n$ is replaced by the value $m(|X|)$. More 
precisely:\\
(ii)' \ $\Omega^{(m)}(X,P)$ is true over the algebra $\mathcal{B}$ \ if and only if the Turing 
machine $P$ accepts input $X$ within time \ $\exp(m(|X|))$; \\ 
(iii)' \ for every $\varepsilon\!>\!0$, there is a constant \ $D\!>\!0$ such that 
the inequalities \ $m(|X|)\!<\!|\Omega^{(m)}(X,P)|\!\leqslant\!D\cdot|P|\cdot m^{2+
\varepsilon}(|X|)$ \ hold for all sufficiently long $X$.

The construction of the new formula remains practically the same; it is only necessary to 
increase the tuples of variables of the 'x' and 'z' types to length $m(|X|)\!+\!1$, i.e., 
instead of $\langle x_{t,0},\ldots,x_{t,n}\rangle$ and $\langle z_{t,0},\ldots,z_{t,n}\rangle
$, where $n\!=\!|X|$, (see Subsection 4.1 in \cite{Lat22} or here) we should respectively use  
$\langle x_{t,0},\ldots,x_{t,m(n)}\rangle$ and $\langle z_{t,0},\ldots,z_{t,m(n)}\rangle$. The 
tuples length of variables the 'q', 'f', and 'd' types  remains the same. We should also make 
appropriate substitutions in other places, except for the "large" conjunction in the formula 
$\chi(0)$ describing the initial configuration (see Subsection 6.2 in \cite{Lat22} or here). In 
the proof, it is also sufficient to restrict ourselves to these substitutions.

Since $Th\mathcal{B}\!\in\!\mathbf{PSPACE}$, most precisely, the truth of the formulae 
$\Omega(X,P)$ and $\Omega^{(m)}(X,P)$ can be checked using additionally only $\mathcal{O}|
(\Omega(X,P)|)$ and $\mathcal{O}|(\Omega^{(m)}(X,P)|)$ memory cells respectively, we, in 
essence, obtain the following corollaries of slightly generalized Theorem \ref{old-main}.

\begin{cor}\label{corEvsPsp} 
(i) The class $\mathbf{E}\!\rightleftharpoons\bigcup_{c\geqslant0}DTIME[\exp(cn)]$ 
is a subclass of $DSPACE[n^3]$.\\
(ii) $\mathbf{EXP}\subseteq\mathbf{PSPACE}$.  \qed
\end{cor}

But it is well known that $\mathbf{PSPACE}\subseteq\mathbf{EXP}$, consequently, 
$\mathbf{PSPACE}=\mathbf{EXP}$. We have already essentially shown (so far informally) 
one of the basic statements about the equality of classes, 
namely, $\mathbf{PSPACE}=\mathbf{EXP}$, using the results of \cite{Lat22}. This allows us to 
obtain by induction on $k$ the equality of the bigger classes $\mathbf{^{(k+1)}EXP}$ and 
$\mathbf{^{k}EXPSPACE}$ by applying the padding argument.

\subsection{Main goals and objectives; the paper's structure}\label{s1.2}

The extension of the simulation period, which we have demonstrated above by considering 
generalized Theorem \ref{old-main} (see also Theorem \ref{Main} below), is a fairly easy task 
concerning its proof, although it results in an impressive consequence that \ 
$\mathbf{^kEXPSPACE}\!=\!\mathbf{^{(k+1)}EXP}$ \ for every $k\!\geqslant\!0$ (Theorem 
\ref{Imp-k}). It is more interesting to learn to write the modeling formulae $\Omega^{(m)}(X,P)$
for the Turing machines having the two tapes because such machines allow to definite 
the space complexity that is less than the length of an input string. In addition, we will 
reduce (not much) the upper bound of the length of these sentences also for $m(n)\!=\!\mathcal{O}
(n)$ by modifying some fragments of the simulation formulae (both these purposes 
will be achieved in Sections \ref{s2},\ref{s4}--\ref{s6}; see also Section \ref{s8}). This will 
allow us to obtain the more precise lower bound of the recognition complexity for the many 
decidable theories, specifically for $TQBF$ and $Th\mathcal{B}$ from Theorem \ref{old-main}. 

Our main goal is to prove that any problem from class $\mathbf{P}$ can be solved using almost 
logarithmic memory. This goal will be achieved in section \ref{s7} (Theorem \ref{PvsLog}), not 
only by limiting the simulation period but mainly by simulating the computations of two-tape 
Turing machines. To understand the proof of this theorem, one needs a detailed description 
of the formulae that model the actions of two-tape machines and which cannot be obtained from 
the formulation of Theorem \ref{Main}. For this reason, a complete proof of the main theorem 
on modeling is given in Sections \ref{s3}--\ref{s6}. At that, the paper's results  
\cite{Lat22} are improved to some extent, as described at the end of Subsection \ref{s2.1}. 

\subsection{Used methods and the main idea}\label{s1.3}

We will use the technique of the immediate codings of the machine actions employing quantified 
Boolean formulae. Let us notice that this language has good enough expressiveness.  

The computation simulation using quantified Boolean formulae was well-known for a long 
time. For example, back in 1971, S.A.~Cook modeled with such formulae the actions of 
nondeterministic machines in polynomial time \cite{Coo71}; this is probably their first 
application for modeling computations. Cook's method gives the polynomially bounded simulating 
$\exists$-formulae that model very cumbersome computations since a nondeterministic machine  
potentially is able to produce an exponential amount of configurations (or \emph{instantaneous 
descriptions} \cite{AHU76,Coo71,SM73}) in polynomial bounded time; and the description of each of 
these potential configurations can be deduced from Cook's formula.

Then, A.R.~Meyer and L.J.~Stockmeyer showed in 1973 that a language $TQBF$ is polynomially 
complete in the class $\mathbf{PSPACE}$ \cite{SM73}. This implies, in particular, that there is 
an algorithm, which produces a quantified Boolean formula for every Turing machine $P$ and for 
any input string $X$ in polynomial time, and all these sentences model the computations of the 
machine $P$ in polynomial space; namely, each of these sentences is true if and only if $P$ 
accepts the given input. The long enough computations are simulated at that, as the polynomial 
constraint on memory allows the machine to run during the exponential time \cite{AHU76,AB09,DK14, 
GJ79}. 

The method, which is employed for the implementation of simulation in \cite{SM73}, permits 
writing down a polynomially bounded formula for the modeling of the exponential quantity of 
the Turing machine steps, provided that one step is described by the formula, whose length is 
polynomial (see the proof of Theorem 4.3 in \cite{SM73}). There, one running step of the machine 
is described by a Boolean $\exists$-formula. It corresponds to the subformula of the 
formula, which has been applied for the modeling of the actions of the nondeterministic Turing 
machine in the proof of $\mathbf{NP}$-completeness of the problem SAT \cite{Coo71} (see also 
the proof of Theorem 10.3 in \cite{AHU76}). Therefore this $\exists$-formula will be termed {\it 
the Cook's method formula}.

The construct of Meyer and Stockmeyer is modified as follows in \cite{Lat22} and here. One 
running step of a machine is described by the more complicated formula that is universal  
in essence (see Remark \ref{remQu} in Section \ref{s4}). 
This complication is caused because Cook's method formula is very long 
--- it is far longer than the amount of the used memory. It has a subformula consisting of one 
propositional (or Boolean) variable $C_{i,j,t}$ --- see, for instance, the proof of 
Theorem~10.3 in \cite{AHU76}. This variable is true if and only if the $i$th cell contains the
symbol $X_j$ of the tape alphabet at the instant $t$. But suppose that each of the first 
$k\!+\!1$ squares of tape contains the symbol $X_0$ at time $t$, and the remaining part of the 
tape is empty. This simple tape configuration is described by the formula that has a fragment  
$C_{0,0,t}\wedge C_{1,0,t}\wedge\ldots\wedge C_{k,0,t}$. Only this subformula is 
$2k+1$ in length without taking the indices into account. It is impossible to abridge this 
record since applying of the universal quantifier to the indices is not allowed within 
the confines of the first-order theory. Thus, to describe the machine actions on the 
exponential fragment of tape, Cook's method formula has to be length still more. 

In \cite{Lat22}, the binary notation of the cell number is encoded by a value set of special 
variables $x_{t,0},\ldots,x_{t,m}$, where $m\!=\!\lceil\log T\rceil$, and $T$ is a 
simulation period, while $t$ denotes a step number. A similar notation will also be used 
in this paper when identifying naturally the value of the variable "true" here with the number 
one in \cite{Lat22} and the value "false" with zero. However, the length of the set of 
variables encoding cell numbers on different tapes will be distinct. This length will be 
equal to \ $s\!=\!\lceil\log n\rceil$ \ for the configurations on the first (input) tape and 
$m(n)\!=\!\lceil\log T\rceil$ for the ones on the second (work) tape, where $n$ is 
the length of an input string (see Subsections \ref{s4.1} and \ref{s6.2} and also 
Section \ref{s7} for further details), provided that the function $m(n)$ is 
space-constructible. We will also apply various types of variables for the 
describing configuration on different tapes: the capital letters with subscripts will be used 
for the first tape and lowercase ones for the second. So $\mathcal{O}(s\lceil\log s\rceil)$ 
various symbols (together with indices) are sufficient for the describing of one cell on the 
first tape, and $\mathcal{O}(n\cdot s\cdot\lceil\log s\rceil)$ ones are doing for the 
representation of the whole input. 

In addition, one can describe an instruction action of the machine, which uses $T\!=\!\exp(m(n))$ 
memory cells, with the aid of a formula that is $\mathcal{O}(m(n)\cdot\lceil\log(m(n))\rceil)$ 
in length. The main idea of such a brief description consists of the following. Firstly, a record 
can change in the only square of the second tape for each running step, while all the records 
remain unchanged on the first (input) tape. Both the heads can only shift by one position, 
although the whole computation can use a significant amount of memory. Secondly, the actions 
executed by the Turing machine on a current step exclusively depend on the records in two 
cells, at which the machine heads are aimed. Therefore, it is enough to describe "the 
inspection" of two squares, the shifts of heads, and the possible change in the only cell on 
the second tape in explicit form. In contrast, the contents of the remaining ones can be 
"copied" by applying the universal quantifier (see the construction of the formula 
$\Delta^{\textit{cop}}(\widehat{u})$ in Subsection \ref{s4.2}). So, we need to explicitly 
describe only three simple actions and not more than five squares on the tapes for every running 
step (in fact, we will only describe explicitly two or three cells --- see Subsections \ref{s3.1} 
and \ref{s4.2}) and the proof of Proposition \ref{Pr2}. 

Thus, by sacrificing the brevity of the description of one tape cell, we obtain a more brief 
description of many other fragments of the modeling formula.    

The denoted locality of the actions of deterministic machines has long been used 
in the modeling of machine computation with the help of formulae, \eg
 Lemma 2.14 in \cite{Sto74} or Lemma 7 in \cite{Vor04}.

In the beginning, all variables will be introduced in great redundancy in order 
to facilitate the proof; namely, the variables will have the first indices $t$ 
from 0 to $T$. Next, many of the variables will be eliminated using the 
methods from \cite{SM73} and \cite{Lat22} --- see Subsections \ref{s4.3.2} and 
\ref{s6.3} for further details.

The description of the initial configuration has a length of \ $\mathcal{O}(n\cdot s\cdot\lceil
\log s\rceil+m(n)\cdot\lceil\log m(n)\rceil)$ if the quantifiers are anew used. The condition 
of the successful termination of computations is very short. Hence, the entire formula, 
which simulates the first $\exp(m(n))$ steps of the computations of the machine 
$P$, will be no more than \linebreak 
$\mathcal{O}(n\cdot s\cdot\lceil\log s)\rceil\!+\!m(n)\cdot\lceil\log m(n)\rceil)
+\mathcal{O}(|P|\cdot m(n)\!+\!m^2(n))\cdot\lceil\log m(n)\rceil])$ \
in length taking into account the indices.  
 
\begin{rem}\label{rem1} 
Since $s$ is equal to $\lceil\log n\rceil$, we need to select $\lceil\log n\rceil\cdot
\lceil\log\lceil\log n\rceil\rceil$ squares for the record of the variables  
$X_{0},\ldots,X_{s}$ for describing one cell at the first tape. But the function 
$\lceil\log\lceil\log n\rceil\rceil$ is not fully space-constructible since it grows as 
$o(\log n)$. However, we do not need to measure out $\lceil\log\lceil\log n\rceil\rceil$ 
cells using the exactly same amount of memory; it is enough only to fill $s$ cells with symbols 
$X$ (this we can do even utilizing exactly $s$ cells because the function $\lceil\log n\rceil$ 
is fully space-constructible); after that, we can number these cells using a little more ones 
than $\lceil\log n\rceil\cdot\lceil\log\lceil\log n\rceil\rceil$. It is described in 
Subsection \ref{s6.4} (Lemma \ref{sp-con}) how one can do this in a more general case. 
\end{rem}

\section{The main results}\label{s2}

Here, we will state the main theorem (Theorem \ref{Main}\,, on simulation) and prove some of 
its consequences, which are the most important in the author's judgment. However, the 
application of this theorem to languages of class $\mathbf{P}$ will be discussed separately in 
Section \ref{s7}. The proof of the theorem on modeling is given in sections 
\ref{s3}--\ref{s6}.

It is assumed that the first tape of all 2-DTMs mentioned in this paper is input tape. This 
means that the machine can read the records on the first tape, but it is forbidden to erase 
and write any symbols on this tape. We recall that a function $f(n)$ is termed \ {\it fully 
space-constructible}, if there is a 2-DTM, which puts a marker exactly in the $f(n)$th square 
on the second (working) tape, and the machine head on this tape does not visit $(f(n)\!+\!1)
$th cell for every input having length $n$ \cite{DK14}. Such functions are sometimes named 
{\it space-constructible} \cite{AHU76}. 
 
\begin{thm}[on modeling]\label{Main}
Let \ $m(n)\!\geqslant\lceil\log n\rceil$ \ be a fully space-constructible function increasing 
almost everywhere, $P$ be a two-tape deterministic Turing machine (2-DTM), and $X$ be an input 
string on the first tape of machine $P$. Then there exists a closed Boolean formula (sentence) 
$\Omega^{(m)}(X,P)$ with the following properties:

(i)\ $\Omega^{(m)}(X,P)$ is true if and only if the 2-DTM $P$ accepts input $X$ within 
time \ $\exp(m(|X|))$; 

(ii) \ there  exist constants $D_{1}$, $D_2$, and $D_w$ such that the inequalities ($n\!=\!|X|
$ here and everywhere below) 
\begin{align}\label{in1}
m(n) \ < \ |\Omega^{(m)}(X,P)| \ \leqslant \ D_1\cdot n\cdot\lceil\log n\rceil\cdot 
\lceil\log\lceil\log n\rceil\rceil + \nonumber\\ 
+D_2\cdot m(n)\cdot\lceil\log m(n)\rceil 
+D_w\cdot[(|P|\cdot m(n)\!+\!m^2(n))\cdot\lceil\log m(n)\rceil] 
\end{align}
hold for every sufficiently long $X$; 

(iii) \ there is a polynomial $g(k)$ such that for all $X$ and $P$, the construction time 
of the sentence $\Omega^{(m)}(X,P)$ is not greater than \ $g(|X|\!+\!|P|\!+\!m(|X|))$, 
while required space for this is $\mathcal{O}(|\Omega^{(m)}(X,P)|)$.
\end{thm}  

{\bf See the proof in Sections \ref{s3}--\ref{s6}}. At first, the simulation formulae will have a 
huge number of "redundant" variables. The constructed formulae will be rewritten in the 
brief form after the correctness of this modeling has been ascertained (see Propositions 
\ref{Pr2}(ii), \ref{Pr3}, and \ref{Pr4}(ii) below); the Meyer and Stockmeyer method is 
substantially used at that.

\begin{rems}\label{rem2}
a) \ The function $g(|X|\!+\!|P|\!+\!m(|X|))$ may depend on $|X|$ in a
not polynomial way, despite $g(k)$ is polynomial because the 
function $m(n)$ may be not polynomial; \eg $m(n)$ can be $k$-fold exponential.

b) \ The upper estimate of the length of the modeling formula has three summands in the 
right part of Inequality \ref{in1}. The first of them, $D_1\cdot n\cdot\lceil\log n\rceil\cdot
\lceil\log\lceil\log n\rceil\rceil$ evaluates the length of the formula that describes an 
input string on the first tape. The next $D_2\cdot m(n)\cdot\log m(n)$ does it for the 
formula, which asserts that the second tape is empty before the beginning of the machine work. 
The third summand does it for the subformula that describes the actions of the machine $P$. 
The summand corresponding to the condition of the successful termination of computations is
absent here, as this subformula is very short; its length  
is not more than $c\cdot m(n)$ for appropriate constant $c$. 
\end{rems}

\begin{cor}[of the main theorem]\label{estTQBF} 
Under the conditions of the theorem, if $n\!\leqslant\!
m(n)\!\leqslant\!Cn$ for some constant $C$, then for every $\varepsilon\!>\!0$, there is 
a constant \ $D_{C}\!>\!0$ such that the inequality 
$|\Omega^{(m)}(X,P)|\!\leqslant\!D_{C}\cdot|X|^{2+\varepsilon}$ \ 
is valid for all long enough $X$; more precisely, \break $|\Omega^{(m)}(X,P)|\!\leqslant\!
D_{C}\cdot|X|^2\cdot\lceil\log|X|\rceil$.
\end{cor}

\begin{proof} 
The latter inequality follows from Inequality \ref{in1} and the condition 
$m(n)\!\leqslant\!Cn$; it is valid for all $X$ such that $|X|\geqslant|P|$. For each  
$\varepsilon\!>\!0$, there is a number $k$ such that an inequality \ 
$(\lceil\log n\rceil\leqslant\!n^{\varepsilon}$ holds for \ $n\!\geqslant\!k$. The first 
inequality is obtained from this and the second one.   
\end{proof}

\subsection{The computational complexity of decidable theories}\label{s2.1}

The following consequences of the main theorem require some concepts and material 
from \cite{Lat22}; we will briefly remind the reader of this information. The function 
$F(n)$, which is monotone increasing on all sufficiently large $n$, is called a {\it limit 
upper bound for the class of all polynomials} (LUBP) if, for any polynomial $p$, there is a 
number $k$ such that the inequality $F(n)\!>\!p(n)$ holds for $n\!\geqslant\!k$, i.e., each 
polynomial is asymptotically smaller than $F$. An obvious example of the limit upper bound 
for all polynomials is a $k$-iterated exponential for every $k\!\geqslant\!1$.

The following proposition is a certain generalization of the corresponding 
assertion from Subsection 4.1 in \cite{Rab99}; it is formulated here to estimate 
the computational complexity of decidable theories. However, it is valid for 
arbitrary languages --- see Proposition 2.1 in \cite{Lat22}.

\begin{prop}[\cite{Lat22}, the generalized Rabin and Fischer method]{}\label{mainPr} 
Let $F$ be a limit upper bound for all polynomials and $\mathcal{T}$ be a theory 
in the signature (or {\it underlying language}) $\sigma$. Suppose that for any $k$-tape DTM 
$P$ and every string $X$ on the input tape of this machine, one can effectively construct a 
sentence $S(P,X)$ of the signature $\sigma$, possessing the following properties:

(i) $S(P,X)$ can be built within time $g(|X|\!+\!|P|)$, where $g$ is a polynomial 
fixed for all $X$ and $P$;

(ii) $S(P,X)$ belongs to $\mathcal{T}$ if and only if the machine $P$ accepts the
input $X$ within $F(|X|)$ steps; 

(iii) there exist constants $D,b,s\!>\!0$ such that the inequalities \ \
$|X|\!\leqslant\!|S(P,X)|\!\leqslant\!D\cdot|P|^b\cdot|X|^s$ \ hold for all 
sufficiently long $X$, and these constants do not depend on $P$.

Then \ $\mathcal{T}\!\notin\!k$-$DTIME[F(D^{-\zeta}\cdot n^{\zeta})]$ holds for 
$\zeta\!=\!s^{-1}$. 
\end{prop}

The following statement improves the lower bound on the recognition complexity of the theory 
of the most straightforward Boolean algebra and the language $TQBF$. 

\begin{cor}\label{comRecB} 
Let $\mathcal{B}$ be a two-element Boolean 
algebra. For every $\varepsilon\!>\!0$, neither $TQBF$ nor $Th\mathcal{B}$
belong to \ 2-$DTIME[\exp(D\cdot n^{\rho})]$ \ for some constant $D\!<\!1$, where 
$\rho\!=\!(2\!+\!\varepsilon)^{-1}$ and $n^{\varepsilon}\!\geqslant\!\log n$.
\end{cor}

\begin{proof} 
It straightforwardly follows from the proposition and Corollary \ref{estTQBF}, because the 
language $TQBF$ and theory $Th\mathcal{B}$ not only are polynomially equivalent, but also each 
sentence of $TQBF$ can be rewritten as an equivalent sentence of $Th\mathcal{B}$ with the 
linear lengthening, and vice versa.   
\end{proof}

\begin{cor}[\cite{Lat22}]\label{mainCor1} 
The class $\mathbf{P}$ is a proper subclass of the class $\mathbf{PSPACE}$. 
\end{cor}

\begin{proof} 
The theory $Th\mathcal{B}$ and the language $TQBF$ do not belong to the class $\mathbf{P}$ 
under the previous corollary. But both belong to the class $\mathbf{PSPACE}$; 
moreover, they are polynomially complete for this class \cite{GJ79,SM73}.   
\end{proof}

Let $\mathfrak{K}$ be a class of the algebraic systems, whose signature $\sigma$ contains the 
symbol of the binary predicate $\sim$, and this predicate is interpreted as an equivalence 
relation on every structure of the class; in particular, $\sim$ may be an equality relation. 
These relations are denoted by the same symbol $\sim$ in all systems of the class.

\begin{defi}[\cite{Lat22}]\label{eqnotr}  
Let us assume that there exists a 
\emph{$\sim$-nontrivial} system $\mathcal{E}$ in a class $\mathfrak{K}$, i.e., 
this structure contains at least two $\sim$-nonequivalent elements.  
Then the class $\mathfrak{K}$ is also termed \emph{$\sim$-nontrivial}. 
A theory $\mathcal{T}$ is named \emph{$\sim$-nontrivial} if it has a 
$\sim$-nontrivial model. When the relation $\sim$ is either the equality or there is  
a formula $N(x,y)$ of the signature $\sigma$ such that the sentence $\exists x,yN(x,y)$  
is consistent with the theory $Th\mathfrak{K}$ (or $\mathcal{T}$, or belongs to 
$Th\mathcal{E}$), and this formula means that the elements $x$ and 
$y$ are not equal, then the term "$\sim$-nontrivial" will be replaced with 
\emph{"equational-nontrivial"}. If it does not matter which equivalence relation is meant, 
we will discuss \emph{equivalency-nontrivial} theory, class, or structure.
\end{defi}
 
Prominent examples of such theories are the theories of pure equality and one 
equivalence relation. Furthermore, nearly all decidable theories mentioned in the surveys 
\cite{ELTT65,Rab99} are nontrivial regarding equality or equivalence.

\begin{thm}\label{Teqnotr} 
Let $\mathcal{E}$, $\mathfrak{K}$, and $\mathcal{T}$ 
accordingly be a $\sim$-nontrivial system, class, and theory of the signature 
$\sigma$, in particular, they may be equational-nontrivial. Then \ 
there exists an algorithm that builds the sentence $\Omega^{(T)}(X,P)$ 
of $\sigma$ for every 2-DTM $P$ and any of its input $X$, where 
$T\!\in\!\{Th\mathcal{E},Th\mathfrak{K},\mathcal{T}\}$; \ this formula  
possesses the properties (i) and (ii) of the sentence $S(P,X)$ from the formulation 
of Proposition \ref{mainPr} for $F(|X|)\!=\!\exp(|X|)$. Moreover, for each  
$\varepsilon\!>\!0$, there is a constant $E_{T,\sigma}$ such that the inequalities \ 
$|X|\!<\!|\Omega^{(T)}(X,P)|\!\leqslant\!E_{T,\sigma}\cdot|P|\cdot|X|^{2+\varepsilon}$ \ 
hold for any long enough $X$ and such that $|X|^{\varepsilon}\!\geqslant\!\log|X|$.
\end{thm}

\begin{proof} 
At first, given $X$ and $P$, one can write a simulating Boolean 
sentence $\Omega^{(|X|)}(X,P)$ applying Theorem \ref{Main}. Then, this formula is 
transformed into the required sentence $\Omega^{(T)}(X,P)$ within polynomial time 
as this does in the proof of Theorem 7.1 in \cite{Lat22}; there, it is shown that the length 
of $\Omega^{(T)}(X,P)$ increases linearly in so doing.  
\end{proof}

\begin{cor}\label{comRecEqNtr} 
The recognition complexity of each equivalency-nontrivial decidable theory $\mathcal{T}$, in  
particular, equational-nontrivial, has the non-polynomial lower bound; more precisely,  \ \
$\mathcal{T}\notin$2-$DTIME[\exp(D_{T,\sigma}\cdot n^{\delta})]$, \ where \ \ 
$\delta\!=\!(2+\varepsilon)^{-1}, \ D_{T,\sigma}\!=\!(E_{T,\sigma})^{-\delta}$, 
and $\varepsilon\!>\!0$ is preassigned.
\end{cor}

\begin{proof} 
It straightforwardly follows from this theorem and Proposition \ref{mainPr}. 
\end{proof}

The lower bounds for the computational complexity of the theories in Corollaries \ref{comRecB} 
and \ref{comRecEqNtr} are somewhat higher than those obtained in \cite{Lat22}, although this is 
difficult to see from their formulations. In Lemma 3.1 from \cite{Lat22}, the factor 
$n^{2+\varepsilon}$ arises as an upper bound for the quantity $(n\cdot\log n)^2=n^2\cdot
\log^2n$ for the estimation of the length of the modeling formula. In constrast, this factor 
evaluates $n\cdot(n\log n)$ in Corollary \ref{estTQBF} and Lemma \ref{lfrOm} when $m(n)=n$. 
Moreover, these evaluations  of the computational complexity are now valid for two-tape machines, 
therefore, they are even higher for one-tape machines.

\subsection{Complexity classes and complete problems}\label{s2.2}

In the previous subsection and \cite{Lat22}, it has proved that the class 
$\mathbf{P}$ is a proper subclass of $\mathbf{PSPACE}$; this is based on the not polynomial 
lower bound on the computational complexity of all the equational-nontrivial theories. Another 
immediate consequence of Theorem \ref{Main} is the following.

\begin{thm}\label{ComplB} 
The language $TQBF$ and the theory $Th\mathcal{B}$ are polynomially complete for the class 
$\mathbf{EXP}$. 
\end{thm}

\begin{proof} 
Theorem \ref{Main} explicitly ensures the presence of a polynomial reduction 
of every problem from the class $\mathbf{EXP}$ to appropriate Boolean 
formulae in the case when $m(n)$ is polynomial. 
\end{proof}

\begin{cor}\label{Imp-0} 
The class $\mathbf{PSPACE}$ equals the class $\mathbf{EXP}$. 
\end{cor}

\begin{proof}
 The language $TQBF$ is polynomially complete for both classes 
according to the theorem and \cite{SM73}. 
\end{proof}

\begin{thm}\label{Imp-k} 
For every $k\!\geqslant\!0$, the classes $k$-fold 
Deterministic Exponential Space and $(k\!+\!1)$-fold Deterministic Exponential Time 
coincide, i.e., \ $\mathbf{^kEXPSPACE}=\mathbf{^{(k+1)}EXP}$. 
\end{thm}

\begin{proof} 
The first variant of proof has already been described at the end of Subsection 
\ref{s1.1}. But Corollary \ref{corEvsPsp}(ii), which is based on an informal 
generalization of Theorem \ref{old-main}, is applied there. Now, one can use Corollary 
\ref{Imp-0} instead of the former and obtain strict proof.

We present the second variant of the proof by directly deducing it from Theorem \ref{Main}.
It is known that if a 2-DTM $M$ works within space bounded by the function $s(n)$, 
then it cannot have more than $s(n)\cdot\exp(Cs(n))$ various configurations on the 
working tape for suitable constant $C$ \cite{AHU76,AB09,DK14,GJ79}, so its running time is 
less than $s(n)\cdot\exp(Cs(n))$, provided that $M$ halts on a given input. Hence, 
the inclusion \ $\mathbf{^kEXPSPACE}\subseteq\mathbf{^{(k+1)}EXP}$ holds. 

Let $\mathcal{L}$ be a language in $\mathbf{^{(k+1)}EXP}$, and $P_0$ be a 2-DTM that 
recognizes this language in time $\exp m(n)$, where $m(n)\!=\!\exp_{k}(cn^d)$ for 
some constants $c,d\!>\!0$. According to Theorem \ref{Main}, one can write down the Boolean 
sentence $\Omega^{(m)}(X,P_0)$ such that this formula simulates the actions of the machine 
$P_0$ on each input $X$; at that, Inequalities (\ref{in1}) are valid for some constants $D_1$, 
$D_2$, and $D_w$, and this formula can be built within space bounded a linear function on its 
length. So, the ascertainment of the $\Omega^{(m)}(X, P_0)$ truth replaces the computations of 
the machine $P_0$ on the input $X$; this substitution is realized within space that has 
the order of $|\Omega^{(m)}(X,P_0)|$, i.e., $\mathcal{O}(\exp_{k}(cn^d))^{2+\varepsilon}$ for 
$n^{\varepsilon}\!\geqslant\!\log n$. The question of whether the sentence $\Omega^{(m)}(X, 
P_0)$ is true is decidable within space of the same order.  
\end{proof} 

\section{Necessary preparations for the proof of the main theorem}\label{s3}

In this section, we specify the restrictions on the used Turing machines,
the characteristics of their actions, and the methods of recording their 
instructions and Boolean formulae. These agreements are essential in   
proving the main theorem on simulation (Theorem \ref{Main}). 

\subsection{On the Turing machines}\label{s3.1}
It is implied that simulated deterministic machines have two tapes and fixed arbitrary finite 
tape alphabet $\mathcal{A}$, which contains at least four symbols: the first of them 
is a designating "blank" symbol, denoted $\Lambda$; the second is a designating 
"start" symbol, denoted $\rhd$; and the last two are the numerals 0,1 (almost 
as in Section 1.2 of \cite{AB09}). The machine cannot erase or write down the $\rhd$ symbol.  
The machine tapes are infinite only to the right as in almost all works cited here because 
such machines can simulate the computations, which is $T$ steps in length on the machines with 
two-sided tapes, in linear time of $T$ \cite{AB09}. 

Before the beginning of the machine work, the first tape contains the start symbol 
$\rhd$ at the most left square (at zeroth cell), a finite non-blank input string $X$ alongside 
$\rhd$, and the blank symbol $\Lambda$ on the rest of its cells, i.e., these squares are 
empty. The cells of the second tape do not contain any symbols except the only $\rhd$ located 
at the zeroth cell. Both heads are aimed at the left ends of the tapes, and the machine is in 
the special starting state $q_{start}\!=\!q_0$. When the machine recognizes an input, it 
enters the accepting state $q_{1}\!=\!q_{acc}$ or the rejecting state $q_{2}\!=\!q_{rej}$.

The first tape serves only to store the input string; the machine cannot erase or write anything 
in the cells of this tape. Therefore, the simulated machines have the {\it three-operand} 
instructions of a kind \ $q_{i}\alpha_1,\alpha_2\!\rightarrow\!q_{j}\beta_1,\gamma,\beta_2$, 
where \ $\alpha_1,\alpha_2,\gamma\!\in\!\mathcal{A}$, \ $\beta_1,\beta_2\!\in\!\{L,R,S\}$ 
(naturally, it is assumed that \ $\mathcal{A}\cap\{L,R,S\}\!=\!\varnothing$); \ $\alpha_1$ and 
$\alpha_2$ are the symbols visible by the heads, respectively, on the first and second tapes; 
\ $\beta_1$ and $\beta_2$ are the commands to shift the corresponding head to the left, or the
right, or stay in the same place; $\gamma$ is a symbol, which should be written down on the 
second tape instead of $\alpha_2$; $q_i$ and $q_j$ are the inner states of the machine 
accordingly before and after the execution of this instruction. 

The Turing machines do not fall into a situation when the machine stops, but its answer 
remains undefined. Namely, they do not try to go beyond the left edges of the tapes, i.e., there 
are no instructions of a form \ $q_{i}\rhd,\alpha_2\!\rightarrow\!q_{k}L,\gamma,\beta_2$ or \  
$q_{i}\alpha_1,\rhd\!\rightarrow\!q_{k}\beta_1,\rhd,L$. The programs also do not
contain the {\it hanging} (or {\it pending}) internal states $q_{j}$, for which $j\!\neq\!
0,1,2$, and there exist instructions of a kind \ $\ldots\rightarrow q_{j}\beta_1,\gamma, 
\beta_2$, but no instructions are beginning with $q_{j}\alpha_1,\alpha_2\rightarrow\ldots$ at 
least for one pair of symbols $\langle\alpha_1,\alpha_2\rangle\!\in\!\mathcal{A}^2$. The 
hanging states are eliminated by adding the instructions of a kind $q_{j}\alpha_1,\alpha_2\!
\rightarrow\!q_{j}S,\alpha_2,S$ for each of the missing pair $\langle\alpha_1,\alpha_2\rangle$ 
of alphabet symbols.

\subsection{On the recording of Boolean formulae}\label{s3.2}
We reserve the following {\it natural language alphabet} $\mathcal{A}_1$ for writing the 
Boolean formulae; it can be different from the alphabet of simulated machines $\mathcal{A}$: \\  
a) the capital and lowercase Latin letters for the indication of the types of 
propositional variables; \ \ 
b) Arabic numerals and commas for the writing of indices;  \ 
c) basic Logical connectives $\neg, \wedge, \vee, \rightarrow$;  \ d) the signs of 
quantifiers $\forall, \exists$;  \ e) auxiliary symbols: (,). 

The priority of connectives or its absence is inessential, as a difference in the length of 
formulae is linear in these cases. 

If necessary (for example, if the alphabet $\mathcal{A}$ does not contain $\mathcal{A}_1$), 
these formulae can be rewritten in the alphabet $\mathcal{A}$ (or even $\mathcal{A}
_0\rightleftharpoons\{0,1\}$) in polynomial time and with a linear increase in length. We 
always suppose this when one says that some machine writes a formula.

We also introduce the following abbreviations and agreements for the 
improvement of perception. The transform of the formula written by applying these 
abbreviations to the one in the natural language increases its length linearly.

Firstly, (1)  in long formulae, the square brackets and (curly) braces are equally applied 
with the ordinary parentheses; \\ 
(2) the connective $\wedge$ is sometimes written as '$\cdot$' (as a rule, between the 
variables or their negations) or as '$\&$' (between the long enough fragments of formulae);\\ 
(3) if $x$ is propositional variable, then $\overline{x}\rightleftharpoons\neg x$.  

Secondly, we  will be useful in the additional (auxiliary) logical connectives, "constants", 
and "relations":\\ 
(4) \ $x\equiv y\:\rightleftharpoons\,\:(x\rightarrow y)\wedge(y\rightarrow x)$; \\
(5) \ $0\:\rightleftharpoons\:x\wedge\neg x=x\cdot\overline{x}$ , \qquad\ $1\:
\rightleftharpoons\: x\vee \neg x=x\vee\overline{x}$\;\footnote{The usage of the signs $\equiv$, 
$0$, and 1 seems more natural than $\sim$ (or $\leftrightarrow$), $\bot$, and $\top$ instead 
of them while simultaneously applying the signs of inequality and addition modulo two.};\\
(6) \ $x\oplus y\:\rightleftharpoons\:\overline{x}\cdot y\vee x\cdot\overline{y}$; obviously, 
that $0\oplus0\equiv0$, $0\oplus1\equiv1$, $1\oplus0\equiv1$, and $1\oplus1\equiv0$, i.e., 
this is the modulo 2 addition;\\
(7) \ $x\!<\!y\:\rightleftharpoons\:x\!\equiv\!0\wedge y\!\equiv\!1$.

Thirdly, the following priority of the connectives and their versions is assumed (in 
descending order of closeness of the connection): $\neg$ (in both forms), $\cdot$, $\oplus$, 
$\equiv$, $\wedge$, $\&$, $\rightarrow$, $\vee$.

With some liberties in notation, we use the signs 0,1, and $<$ for both numbers and formulae. 
But this should not lead to confusion since it will always be clear from the context what is 
meant.

The record \ $\widehat{x}$ signifies an ordered set $\langle x_{0},\ldots,x_{s}\rangle$, 
whose length is fixed. Naturally, that the "formula" \ 
$\widehat{x}\!\equiv\!\widehat{\alpha}$ denotes the system of equivalences \ 
$x_{0}\!\equiv\!\alpha_{0}\!\wedge\ldots\wedge x_{s}\!\equiv\!\alpha_{s}$. The record \ 
$\widehat{\alpha}\!<\!\widehat{\beta}\rightleftharpoons\langle\alpha_{0},\ldots,\alpha_{s}
\rangle\!<\!\langle\beta_{0},\ldots,\beta_{s}\rangle$ \ denotes the comparison of tuples in 
lexicographic ordering, i.e., it is the following formula:
\begin{align*}
\alpha_{0}\!<\!\beta_{0}\vee\Bigl\{\alpha_{0}\!\equiv\!\beta_{0}\wedge
\Bigl[\alpha_{1}\!<\!\beta_{1}\vee\Bigl(\alpha_{1}\!\equiv\!\beta_{1}\wedge\bigl
\{\alpha_{2}\!<\!\beta_{2}\vee \\
\vee\bigl[\alpha_{2}\!\equiv\!\beta_{2}
\wedge\bigl(\alpha_{3}\!<\!\beta_{3}\vee\{\alpha_3\equiv\beta_3\wedge\ldots\}\bigr)\bigr]\big
\}\Bigr)\Bigr]\Bigr\}.
\end{align*}  

The tuples of variables with two subscripts will occur only in the form where 
the first index is fixed, for instance, $\langle u_{k,0},\ldots,u_{k,s}\rangle$, and 
we will denote it by \ $\widehat{u}_{k}$.

A binary representation of a natural number $t$ is denoted by $(t)_{2}$. On the other hand, 
if $\widehat{\gamma}$ is a tuple consisting of 0 and 1, then $(\widehat{\gamma})$ will denote 
the binary representation of some natural number $t$. So, $((t)_2)\!=\!t$ and 
$((\widehat{\gamma}))_2\!=\!(\widehat{\gamma})$ according to our dessignations.
Here, 0 and 1 can denote the above-defined formulae or sometimes the numerals.

It is known that if \ $t\!=\!(\widehat{\gamma})\!\rightleftharpoons\!\langle\gamma_{0}
\ldots\gamma_{s}\rangle_2$ \ is a binary representation of a natural number $t$, then the 
numbers \ $t\!\pm\!1$ for \ $0\!\leqslant\!t\!-1\!\!<t\!+\!1\!\leqslant\!\exp(2,s)$ \ will be 
expressed as 
\[
((\widehat{\gamma})\pm 1)_2 =(\widehat{\gamma}\oplus\widehat{v})\ \rightleftharpoons \
\langle\gamma_{0}\!\oplus v_{0},\ldots,
\gamma_{s-2}\!\oplus v_{s-2},\gamma_{s-1}\!\oplus v_{s-1},\,\gamma_{s}\!\oplus\!v_s\rangle_2,
\]
where the corrections $v_i$ are calculated in descending order of indices according to the 
rules defined by the following formulae: 
\begin{eqnarray}\label{pl1}
\kappa_R(\widehat{\gamma},\widehat{v}) \ \rightleftharpoons \ (v_s\!\equiv\!1)\ \wedge \  
(v_{s-1}\!\equiv\!\gamma_s) \ \wedge \   
(v_{s-2}\!\equiv\!\gamma_{s-1}\!\cdot\!v_{s-1})\ \wedge\ldots\wedge \nonumber\\
\wedge\ (v_{1}\!\equiv\!\gamma_2\!\cdot\!v_2)\ \wedge\:(v_{0}\!\equiv\!\gamma_1\!\cdot\!v_1) 
\end{eqnarray}           
for writing \ \ $t\!+\!1$; \ \ and 
\begin{eqnarray}\label{min1}
\kappa_L(\widehat{\gamma},\widehat{v}) \ \rightleftharpoons \ v_s\!\equiv\!1\ \wedge  
\ v_{s-1}\!\equiv\!\overline{\gamma}_s\ \wedge\  
v_{s-2}\!\equiv\!\overline{\gamma}_{s-1}\!\cdot\!v_{s-1}\ \wedge\ldots\wedge \nonumber \\ 
\wedge \ v_{1}\!\equiv
\!\overline{\gamma}_2\!\cdot\!v_2\ \wedge\ v_{0}\!\equiv\!\overline{\gamma}_1\!\cdot\!v_1 
\end{eqnarray}            
for writing \ \ $t\!-\!1$. 

\begin{lem}\label{ll2} 
(i) $|\widehat{\alpha}\!<\!\widehat{\beta}| = \mathcal{O}(|\widehat{\alpha}|+|\widehat{\beta}|)$.

(ii) The binary representations of the numbers $t\!\pm\!1$ can be written in linear time on 
$|(t)_{2}|$; hence their lengths (together with the formula $\kappa_R(\widehat{\gamma},
\widehat{v})$ or $\kappa_L(\widehat{\gamma},\widehat{v})$) are $\mathcal{O}(|(t)_{2}|)$.  
\end{lem}

\begin{proof} 
It is obtained by direct calculation.   
\end{proof}

\section{The beginning of the construction of modeling formulae}\label{s4}

We recall that $\mathcal{A}$ is the alphabet of simulated machines. Before writing a 
formula $\Omega^{(m)}(X,P)$, we add $2|\mathcal{A}|^2$ {\em instructions of the idle run} to a 
program $P$; they have the form \ $q_k\alpha_1,\alpha_2\to q_kS,\alpha_2,S$, where \ $k\!\in\!
\{1(accept),$ $2(reject)\}, \ \alpha_1,\alpha_2\!\in\!\mathcal{A}$. While the machine executes 
them, the tape configurations do not change.

\subsection{The primary (basic) and auxiliary variables}\label{s4.1} 
To simulate the operations of a Turing machine $P$ on an input $X$ 
within the first \ $T\!=\!\exp(m(|X|))$ \ steps, it is enough to describe its actions on a  
zone, of the second tape, which is $T\!+\!1$ squares in width since if $P$ starts 
its run in the zeroth cell, it can finish a computation at most in the $T$th one.
Since the record of the number $(T)_2$ has $m\!+\!1\!=\!m(|X|)\!+\!1$ bits, the cell 
numbers of the second tape are encoded by the values of the ordered sets of the variables 
of a kind "$x$": \ $\widehat{x}_{t}\!=\!\langle x_{t,0},\ldots, x_{t,m}\rangle$,  
which have a length of \ $m\!+\!1$. The first 
index $t$, i.e., {\it the color} of the record, denotes the step number, after 
which a configuration under study appeared on the second tape. So the formula \ 
$\widehat{x}_t\!\equiv\!\widehat{\alpha}\rightleftharpoons x_{t,0}\!\equiv\!
\alpha_{0}\wedge\ldots\wedge x_{t,m}\!\equiv\!\alpha_{m}$ \ assigns the 
number \ $(\widehat{\alpha})$ \ of the needed cell of this tape in the binary notation 
at the instant $t$.

The formula \ $\widehat{X}\!\equiv\!\widehat{\gamma}\rightleftharpoons X_0\!\equiv\!
\gamma_0\wedge\ldots\wedge X_{s}\!\equiv\!\gamma_{s}$ \ assigns the number \ 
$(\widehat{\gamma})$ \ of the cell of the first tape in the binary notation, where 
$s=\lceil\log|X|\rceil$. It is clear that \ $\exp(s+1)\!>\!|X|\!+\!1$, which is normal 
since the first head can go beyond the right edge of the input string in the computation process;  
on the other hand, we propose that this head does not visit the $(|X|\!+\!2)$th cell. 
Because nothing can be erased or written in the cells of the first tape, one does not need to 
indicate the step number here.

Let $r$ be so great a number that one can simultaneously write down all the state numbers of 
$P$ and encode all the symbols of the alphabet $\mathcal{A}$ through the bit combinations 
in the length $r\!+\!1$. Thus, $\exp(r\!+\!1)\!\geqslant\!|\mathcal{A}|\!+\!U$, 
where $U$ is the maximal number of the internal states of $P$; and if $\beta
\!\in\!\mathcal{A}$, then $c\beta\rightleftharpoons\langle c\beta_{0},\ldots,c\beta_{r}
\rangle$ will be the $(r\!+\!1)$-tuple, which consists of 0 and 1 and encodes $\beta$.

The formula $\widehat{f}_{t}\!\equiv\!c\varepsilon$ represents a record of the symbol 
$\varepsilon\!\in\!\mathcal{A}$ in some cell of second tape after step $t$, here $\widehat{f}
_{t}$ is the $(r+\!1)$-tuple of variables. When the $(\widehat{\mu})$th cell of the second 
tape contains the $\varepsilon$ symbol after step $t$, then this fact is associated with the 
{\it quasi-equation} (or {\it the clause}) of color $t$:
\begin{multline*} 
\psi_{2,t}(\widehat{\mu}\!\rightarrow\!\varepsilon)\:\rightleftharpoons\:
[\widehat{x}_{t}\!\equiv\!\widehat{\mu}\rightarrow
\widehat{f}_{t}\!\equiv\!c\varepsilon]\: \rightleftharpoons\\
\rightleftharpoons 
(x_{t,0}\!\equiv\!\mu_{0}\!\wedge\!\ldots\!\wedge\!x_{t,m}\!\equiv\!\mu_{m})
\!\rightarrow (f_{t,0}\!\equiv\!c\varepsilon_{0}\!\wedge\!\ldots\!\wedge\!
f_{t,r}\!\equiv\!c\varepsilon_{r}). 
\end{multline*}

The clause
\[ 
\psi_{1}(\widehat{\zeta}\!\rightarrow\!\varepsilon)\:\rightleftharpoons\:
[\widehat{X}\!\equiv\!\widehat{\zeta}\rightarrow
\widehat{F}\!\equiv\!c\varepsilon]\: \rightleftharpoons
(X_{0}\!\equiv\!\zeta_{0}\!\wedge\!\ldots\!\wedge\!X_{s}\!\equiv\!\zeta_{s})
\!\rightarrow (F_{0}\!\equiv\!c\varepsilon_{0}\!\wedge\!\ldots\!\wedge\!
F_{r}\!\equiv\!c\varepsilon_{r}) 
\]
corresponds to the entry of the $\varepsilon$ symbol at $(\zeta)$th cell on the first tape.

The tuples of variables $\widehat{q}_{t}$, $\widehat{D}_t$, and $\widehat{d}_{t}$ are 
accordingly used to indicate the number of the machine's internal state and 
the codes of the symbols visible by the heads on the first and second tapes at the instant 
$t$; the tuples of variables $\widehat{Z}_t$ and $\widehat{z}_{t}$ store the scanned cell's  
numbers on the first and second tapes. For every step $t$, a number \ 
$i\!=\!(\widehat{\delta})$ of the machine state $q_{i}$ 
and the scanned square's numbers $(\widehat{\zeta})$ and $(\widehat{\xi})$ together with the 
symbols $\alpha$ and $\beta$, which are contained there, are represented by a united
\emph{$\pi$-formula} of color $t$:  
\begin{multline*}
\pi_{t}((i)_{2},\alpha,\beta,\widehat{\zeta},\widehat{\xi}) \
\rightleftharpoons \ [\widehat{q_{t}}\!\equiv\!\widehat{\delta}\wedge\widehat{D}_{t}\! 
\equiv\!c\alpha\wedge\widehat{d}_t\equiv\widehat{\beta}\wedge\widehat{Z}_t\equiv
\widehat{\zeta}\wedge\widehat{z_{t}}\!\equiv\!\widehat{\xi}] \ \rightleftharpoons \\
\rightleftharpoons(q_{t,0}\!\equiv\!\delta_{0}\wedge\ldots\wedge q_{t,r}\!\equiv\!\delta_{r})
\wedge(D_{t,0}\!\equiv\!c\alpha_{0}\wedge\ldots\wedge D_{t,r}\!\equiv\!c\alpha_{r})\wedge\\ 
\wedge(d_{t,0}\!\equiv\!c\beta_{0}\wedge\ldots\wedge d_{t,r}\!\equiv\!c\beta_{r})\wedge
(Z_{t,0}\!\equiv\!\zeta_{0}\wedge\ldots\wedge Z_{t,s}\!\equiv\!\zeta_{s})\:\wedge 
(z_{t,0}\!\equiv\!\xi_{0}\wedge\ldots\wedge z_{t,m}\!\equiv\!\xi_{m}).
\end{multline*}
This formula expresses a condition for the applicability of instruction \ 
$q_i\alpha,\beta\!\rightarrow\!\ldots$; in other words, this is a {\it timer} that 
"activates" exactly this instruction, provided that the heads scan the $(\widehat{\zeta})$th 
and $(\widehat{\xi})$th cells.

The tuples of basic variables \ $\widehat{x}_t,\ \widehat{z}_{t}$, $\widehat{X}, \ 
\widehat{Z}_t$, $\widehat{q}_{t},\ \widehat{F}, \ \widehat{f}_{t}, \ \widehat{D}_t$, and 
$\widehat{d}_{t}$ are introduced in great abundance to facilitate the proof. But a 
final modeling formula will contain only those for which $t\!=\!0$ or \ $t\!=\!T 
\rightleftharpoons\exp(m(n))$. 

The sets of basic variables have different lengths. 
However, this will not lead to confusion since the tuples of the first two types ($x$ and $z
$) will always be $m+\!1$ in length without indices, the second two types ($X$ and $Z$) will 
always be $s\!+\!1$ in length, whereas the last five ones ($q, \ F, \ f, \ D$, and $d$) will 
have a length of $r\!+\!1$. The sets of "constants" or other variables may also be different 
in length. But such tuples will always be unambiguously associated with some of 
the primary ones by using the connective $\equiv$.

The other variables are auxiliary. They will be described as needed. Their task 
consists of determining the values of the basic variables of the color $t+\!1$ provided 
that the primary ones of the color $t$ have the "correct" values; namely, this attribution 
should adequately correspond to that instruction, which is employed at step \ $t\!+\!1$. 

\begin{lem}\label{lcl} 
If the indices are left out of the account, then   
$|\psi_{2,t}(\widehat{u}\!\rightarrow\!\beta)|=\mathcal{O}(m+r)$,  
$|\psi_1(\widehat{V}\rightarrow\alpha)|=\mathcal{O}(s+r)$, and a timer ($\pi$-formula) 
will be \ $\mathcal{O}(m\!+\!r+s)$ in length. 
\end{lem}

\begin{proof} 
It is obtained by direct counting.  
\end{proof}

\subsection{The description of the instruction actions}\label{s4.2} 

To improve perception and to facilitate proof, formulae describing the action of one instruction 
and the actions of the entire Turing machine in one step are constructed with an alternation of 
universal and existential quantifiers in the subformula $\Delta^{cop}(\widehat{u})$. Still, this 
alternation is not essential, as noted in subsection 4.3 --- see Remark \ref{remQu}.

The following formula $\varphi(k)$ describes the actions of the $k$th instruction \break
 $M(k)\!=\!q_{i}\alpha_1,\alpha_2\!\rightarrow\!q_{j}\beta_1,\gamma,\beta_2$ (including the 
idle run's instructions; see the beginning of this section) at some step $t\!+\!1$, where 
$\alpha_1,\alpha_2,\gamma\!\in\!\mathcal{A}$ and \ $\beta_1,\beta_2\!\in\!\{R,L,S\}$:  
\begin{align*}
\varphi(k) \ \rightleftharpoons \ \forall\,\widehat{u},\!\widehat{U}\;\forall\,
\widehat{v},\!\widehat{V}\:\forall\,\widehat{h},\!\widehat{H}\: 
\Bigl\{\Bigl[\pi_{t}((i)_{2},\alpha_1,\alpha_2,\widehat{U},\widehat{u}) \ \& 
\Bigl(\kappa_1(\beta_1)(\widehat{U},\widehat{V}) \ \& \ 
\kappa_2(\beta_2)(\widehat{u},\widehat{v})\Bigr) \ \& \\ \& \ 
\Gamma^{ret}(\beta_1,\beta_2)\Bigr] \quad \rightarrow\ 
 \quad \bigl\{\Delta^{cop}(\widehat{u})
 \ \& \ \Delta^{wr}(\widehat{u})\ \& \  
\pi_{t+1}\bigl((j)_2,\widehat{H},\widehat{h},\widehat{U}(\beta_1),\widehat{u}(\beta_2)
\bigr\}\Bigr\}.
\end{align*}             

We now will describe the subformulae of $\varphi(k)$ with the free main variables  
$\widehat{q}_{t},\widehat{X},$ $\widehat{x}_{t},\widehat{Z}_t,\widehat{z}_{t},\widehat{D}_t,
\widehat{d}_{t},\widehat{F},\widehat{f}_{t},\widehat{q}_{t+1},\widehat{x}_{t+1},
\widehat{Z}_{t+1},\widehat{z}_{t+1},\widehat{D}_{t+1}$, $\widehat{d}_{t+1}$,  
and $\widehat{f}_{t+1}$.

The first $\pi$-formula of color $t$ plays the role of \emph{a timer} and "starts 
up the fulfillment" of the instruction with the prefix \ 
"$q_{i}\alpha_1,\alpha_2\!\rightarrow$" \ provided that the heads scan the $(\widehat{U})$th 
cell on the first tape and the $(\widehat{u})$th square on the second. In this case, 
the numbers $(\widehat{U}(\beta_1))$ and $(\widehat{u}(\beta_2))$ of the cells, which the 
heads will view after the execution of the instruction $M(k)$, can accordingly be 
found as $(\widehat{U}(\beta_1))\rightleftharpoons(\widehat{U}\oplus\widehat{V})$ 
and $(\widehat{u}(\beta_2))\rightleftharpoons(\widehat{u}\oplus\widehat{v})$, where the tuples 
of corrections $\widehat{V}$ and $\widehat{v}$ satisfy the conditions $\kappa_1(\beta_1)
(\widehat{U},\widehat{V})$ and $\kappa_2(\beta_2)(\widehat{u},\widehat{v})$ for given 
meta-symbols $\beta_1,\beta_2$. These conditions almost coincide with the formulae \ref{pl1} 
and \ref{min1} in Subsection \ref{s3.2}\;; more precisely:
\begin{equation}\label{LR1}
      \kappa_1(\beta_1)(\widehat{U},\widehat{V})\rightleftharpoons\begin{cases}
      \kappa_{\beta_1}(\widehat{U},\widehat{V}),    &\text{if} \ \ \beta_1\!\in\!\{L,R\};\\
      \widehat{V}\equiv\widehat{0},                &\text{if} \ \beta_1=S ;       
      \end{cases}
\end{equation}
and
\begin{equation}\label{LR2}
      \kappa_2(\beta_2)(\widehat{u},\widehat{v})\rightleftharpoons\begin{cases}
      \kappa_{\beta_2}(\widehat{u},\widehat{v}),    &\text{if} \ \ \beta_2\!\in\!\{L,R\};\\
      \widehat{v}\equiv\widehat{0},                &\text{if} \ \beta_2\!=\!S.       
      \end{cases}
\end{equation}

The informal sense of subformula $\Gamma^{ret}(\beta_1,\beta_2)$ is the following: it 
"seeks" the codes $\widehat{H}$ and $\widehat{h}$ of the symbols $\lambda_1$ and $\lambda_2$, 
which will be scanned after the current step $t\!+\!1$ (for this reason, it is named 
"retrieval"). For this purpose, it "inspects" the squares that are to the right or left of the 
cells $(\widehat{U})$ and $(\widehat{u})$ (if \ $\beta_1,\beta_2\!\in\!\{R,L\}$) or it nothing 
seeks on the $l$th tape (when $\beta_l\!=\!S$): \ $\Gamma^{ret}(\beta_1,\beta_2)  
\rightleftharpoons \Gamma^{ret}_1(\beta_1) \wedge \Gamma^{ret}_2(\beta_2)$, where \ 
$\Gamma^{ret}_1(\beta_1)\rightleftharpoons\psi_1(\widehat{U}(\beta_1)\!\rightarrow\!             
\widehat{H})$ for each $\beta_1$, \  $\Gamma^{ret}_2(\beta_2)\rightleftharpoons\psi_{2,t}
(\widehat{u}(\beta_2)\!\rightarrow\!\widehat{h})$ \ for $\beta_2\!\in\!\{R,L\}$, and \ 
$\Gamma^{ret}_2(S)\rightleftharpoons\widehat{h}\!\equiv\!c\gamma$, as the symbol $\lambda_2\!=
\!\gamma$ will be scanned by the second head in this case. 
Indeed, when $\beta_1\!=\!S$, then $(\widehat{U}(\beta_1))\!=\!(\widehat{U}\oplus\widehat{V})
\!=\!(\widehat{U})$ and $\widehat{H}\!=\widehat{\alpha}_1$, while if $\beta_2\!=\!S$, then $
(\widehat{u}(\beta_2))\!=\!(\widehat{u}\oplus\widehat{v})\!=\!(\widehat{u})$ and $\widehat{h}
\!=c\gamma\!=\!\lambda_2$.

The formula \ $\Delta^{cop}(\widehat{u})$ changes the color of records in all the 
cells of the second tape, whose numbers are different from $(\widehat{u})$; 
in other words, it "copies" the majority of records on this tape:
\[
\Delta^{cop}(\widehat{u}) \ \rightleftharpoons \ 
\forall\,\widehat{w}\,[\neg(\widehat{w}\!\equiv\!\widehat{u})\rightarrow
\exists\,\widehat{g}(\psi_{2,t}(\widehat{w}\!\rightarrow\!\widehat{g})\wedge
\psi_{2,t+1}(\widehat{w}\!\rightarrow\!\widehat{g}))].
\]
 
The formula \ $\Delta^{wr}(\widehat{u})$ "puts" the symbol $\gamma$ 
of color $t\!+\!1$ in the $(\widehat{u})$th square on the second tape: \ 
$\Delta^{wr}(\widehat{u})\rightleftharpoons\psi_{2,t+1}(\widehat{u}\!\rightarrow\!\gamma).$

The second $\pi$-formula of the color $t+\!1$ "aims" the heads at the $(\widehat{U}\oplus
\widehat{V})$th and $(\widehat{u}\oplus\widehat{v})$th 
cells; "places" the symbols $\widehat{H}$ and $\widehat{h}$ in these 
locations; and "changes" the machine state's number on $j$:  
\[
\widehat{Z}_{t+1}\!\equiv\!\widehat{U}\!\oplus\!\widehat{V}\ \wedge \ \widehat{z}_{t+1}\!
\equiv\!\widehat{u}\oplus\widehat{v} \ \wedge \ \widehat{D}_{t+1}\!\equiv\!\widehat{H} \
\wedge \ \widehat{d}_{t+1}\!\equiv\!\widehat{h} \ \wedge \ \widehat{q}_{t+1}\!\equiv\!(j)_2.
\]

\begin{lem}\label{lphi} 
If the indices are left out of account, then $|\varphi(k)|=\mathcal{O}(m+r+s)$. 
\end{lem}

\begin{proof} 
This follows from Lemmata \ref{ll2} and \ref{lcl} by immediate calculation. 
\end{proof}

\subsection{The description of the running steps and configurations}\label{s4.3} 
The Meyer and Stockmeyer technique will be applied to describe the machine's actions over an 
exponential period.
 
\subsubsection{One step.} Let $N$ be a number of the machine's instructions $P$ together 
with $2|\mathcal{A}|^2$ instructions of the idle run (see the beginning of this section). The 
formula $\Phi^{(0)}(P)$ that describes one step (whose number is $t\!+\!1$) of the 
machine $P$ run is of the form:
\begin{equation}\label{Phi0}
\Phi^{(0)}(P)(\widehat{y}_{t},\widehat{y}_{t+1}) \ \rightleftharpoons \
\bigwedge_{0<\,k\leqslant\,N}\varphi_{}(k)(\widehat{y}_{t},\widehat{y}_{t+1}),
\end{equation}  
where \
$\widehat{y}_{t}\rightleftharpoons\langle\widehat{q}_{t},\widehat{X},\widehat{x}_{t},
\widehat{Z}_{t},\widehat{z}_{t},\widehat{D}_{t},\widehat{d}_{t},\widehat{F},\widehat{f}_{t}
\rangle$ \ \ and \ \
$\widehat{y}_{t+1}\rightleftharpoons\langle\widehat{q}_{t+1},\widehat{X}$, $\widehat{x}_{t
+1},\widehat{Z}_{t+1}$, $\widehat{z}_{t+1},\widehat{D}_{t+1},$ $\widehat{d}_{t+1},\widehat{F},
\widehat{f}_{t+1}\rangle$ are two \ $(2m\!+\!5r\!+\!2s)$-tuples of its free variables.

\begin{rems}\label{remQu} 
\textbf{a)}. Let $\langle\chi\rangle$ be a quantifier-free part of a  
formula $\chi$. Using the well-known Tarski and Kuratowski algorithm, we obtain  
that the prenex-normal form of $\varphi(k)$ is  $\forall\,\widehat{U},\widehat{u},\widehat{V},
\widehat{v}\,\forall\widehat{H},\widehat{h}\,\forall\,
\widehat{w}\,\exists\,\widehat{g}\langle\varphi(k)\rangle$. 
Recall that the variables $\widehat{U},\widehat{u},\widehat{V},\widehat{v}$, $\widehat{H},
\widehat{h},\widehat{w}$, 
and $\widehat{g}$ are auxiliary and "attached" to the corresponding basic 
variables (see Subsection 4.1). The variables $\widehat{g}$ "service" only the 
variables $\widehat{f}_t$ and $\widehat{f}_{t+1}$. Hence, the tuple 
$\widehat{g}$ has a length of \ $r\!+\!1$, but not \ $m\!+\!1\!=\!m(n)\!+\!1$ and not \ 
$s\!+\!1\!=\!\lceil\log n\rceil\!+\!1$, i.e., its length does not depend on $n\!=\!|X|$. 
Therefore, we can substitute the subformula \ $\Delta^{cop}(\widehat{u})$ \ with \ 
$\Delta^{cop}_1(\widehat{u})\rightleftharpoons\forall\widehat{w}
\bigvee_{\rho\in\mathcal{A}}\langle\Delta^{cop}(\widehat{u})(\rho)\rangle,$ where all 
possible tuples $c\rho$ are placed instead of the tuples $\widehat{g}$,  
since the sense of the basic variables $\widehat{f}_t$ and $\widehat{f}_{t+1}$ is "storing" 
information about the alphabet symbols. It is clear that \ $|\Delta^{cop}_1(\widehat{u})|
\leqslant|\mathcal{A}|\cdot|\Delta^{cop}(\widehat{u})|$ because one does not need to 
write a prefix $\exists \widehat{g}$. Thus, the formulae $\varphi(k)$ and 
$\Phi^{(0)}(P)(\widehat{y}_{t},\widehat{y}_{t+1})$ \ are universal in essence.

\textbf{b)}. \label{Ind} 
The formula $\varphi(k)$ does not depend on concrete entries on 
the tapes at instant $t$ but depends only on the length of the input string $n$ (since 
$s+1=\lceil\log n\rceil+1$ is the number of variables $Z_{t,i}, \ X_i$, and $Z_{t+1,i}$) and 
the value of function $m(n)$ (which equals to the number of variables $x_{t,j}, \ z_{t,j}, \ 
z_{t+1,j},$ and $x_{t+1,j}$) at this concrete value $n$. Therefore the whole formula 
$\Phi^{(0)}(P)$ also depends only on the parameter $n$ and the function $m$ for fixed program 
$P$. There is a full analogy with programs for real computers, which can contain some 
parameters depending on the input length. For example, a program for multiplying square 
matrices in programming languages Algol, BASIC, Pascal, $C\#$, $C^{++}$, etc. can contain 
parameter $n$, which is the number of rows and entered additionally as a rule (certainly, this 
program is a learning exercise). However, a competently written program should not depend on 
specific values of elements of matrices to be multiplied.
\end{rems}

\begin{lem}\label{trcl} 
(i) If \ $\widehat{x}_{t}\!\neq\!\widehat{\mu}$, then 
a clause $\psi_{2,t}(\widehat{\mu}\!\rightarrow\!\varepsilon)$ will be true 
independently of the values of variables $\widehat{f}_{t}$. In particular, any 
quasi-equation, which is contained in the record of 
$\langle\Delta^{cop}(\widehat{u})\rangle$, will be true, if its 
color is $t$ or $t+\!1$, and at the same time \ $\widehat{x}_{t}\!\neq\!\widehat{w}$ 
or \ $\widehat{x}_{t+1}\!\neq\!\widehat{w}$ respectively. 

(ii)  When \ $\widehat{X}\!\neq\!\widehat{\mu}$, then 
a clause $\psi_{1}(\widehat{\mu}\!\rightarrow\!\varepsilon)$ will be true 
independently of the values of variables $\widehat{F}$.

(iii) For some constant $C_{1}$, the inequality \
$|\Phi^{(0)}(P)(\widehat{y}_{t},\widehat{y}_{t+1})|\!\leqslant\!C_{1}\cdot
|P|\cdot|\varphi(N)|$ holds, provided that the program $P$ is not empty.
\end{lem}

\begin{proof}
(i),(ii) The premises of clauses are false in these cases.\\
(iii) If \ $N\!-\!2|\mathcal{A}|^2\neq\!0$ (see the definition-equation \ref{Phi0} and the 
beginning of this section), then $N\cdot\lceil\log N\rceil\!<\!C_{2}\cdot|P|$; here $\log N$ 
locations are required to number (using indices $k$ for $1\!\leqslant\!k\!\leqslant\!N$) the 
tuples of auxiliary variables $\widehat{u},\!\widehat{U},\widehat{v},\widehat{V},\widehat{h},
\widehat{H},\ldots$ entered in the record of each formula $\varphi(k)$.  
This implies the lemma assertion.  
\end{proof}

\subsubsection{The exponential amount of steps}\label{s4.3.2}
The formulae  $\Phi^{(k)}(P)(\widehat{y}_{t},\widehat{y}_{t+e(k)})$ 
conform to the actions of machine $P$ over a period of time \ 
$e(k)\!\rightleftharpoons\!\exp(k)$. They are defined by induction on $k$:
\[
\Phi^{(k+1)}(P)\!\rightleftharpoons\!\exists\,\widehat{Y}\,
\forall\,\widehat{a}\,\forall\,\widehat{b}\bigl\{\bigl[(\widehat{y}_{t}\!
\equiv\!\widehat{a}\wedge\widehat{Y}\!\equiv\!\widehat{b})\vee
(\widehat{Y}\!\equiv\!\widehat{a}\wedge
\widehat{b}\!\equiv\!\widehat{y}_{t+e(k+1)})\bigr] \rightarrow 
 \ \Phi^{(k)}(P)(\widehat{a},\widehat{b})\bigr\},
\] 
where \ $\widehat{Y},\widehat{a},\widehat{b}$ are the $(2m\!+\!5r\!+\!2s)$-tuples  
of the new auxiliary variables.

\subsubsection{The description of configurations}\label{s4.3.3} 
Let $L_2(t)$ be a 
configuration recorded on the second tape after step $t$ (it may be unrealizable). Namely, 
every cell, whose number is $(\widehat{\mu})$, contains a symbol $\varepsilon(\widehat{\mu})$; 
the scanned squares have the numbers $(\widehat{\zeta}_1)$ and $(\widehat{\zeta}_2)$; and a 
machine is ready to execute an instruction \ $q_{i}\alpha_1,\alpha_2\!\rightarrow\!\ldots$. 
Then the following formula corresponds to this configuration: 
$$ \Psi L_2(t)(\widehat{y}_{t}) \ \rightleftharpoons \
\pi_{t}((i)_{2},\alpha_1,\alpha_2,\widehat{\zeta}_1,\widehat{\zeta}_2) \ \& \
\bigwedge_{0\leqslant\,(\widehat{\mu})\leqslant\,T}
\psi_{2,t}(\widehat{\mu}\!\rightarrow\!\varepsilon(\widehat{\mu})),$$
we recall that $T\!=\!\exp(m)$. The variables $\widehat{X}$ and $\widehat{F}$ are dummies 
here. Let us notice that this formula also indicates the location of the head on the first 
tape and describes that symbol, which is in this cell.

It remains only to describe the entries in the cells of the first tape: 
\[ 
\Psi K_1(\widehat{X},\widehat{F})\rightleftharpoons \psi_1(\widehat{0}\rightarrow
\rhd) \ \& \bigwedge_{1\leqslant\,(\widehat{\xi})\leqslant\,n}
\psi_{1}(\widehat{\xi}\!\rightarrow\!\chi(\widehat{\xi})) \ \& 
\quad \forall\widehat{u}_0\:\bigl(\widehat{u}_0>(n)_2 \ \rightarrow \ \psi_1(\widehat{u}
_0\rightarrow\Lambda)\bigr),
\]   
here the letter "L" is replaced by the letter "K" because these records on the first tape are 
undoubtedly real. The meaning of this entire formula is as follows: the $\rhd$ symbol is 
located in the zeroth cell, the input string $X=\langle\chi_1,\ldots,\chi_n\rangle$ is written 
next to $\rhd$, and the remaining cells of the tape are empty.%

The configuration $C(t)=K_1\cup L_2(t)$ of both tapes together at instant $t$ is described by 
the formula 
\[
\Psi(C)(\widehat{y}_{t}) \ \rightleftharpoons \ \Psi K_1(\widehat{X},\widehat{F}) \ \& \ 
\Psi L_2(t)(\widehat{y}_{t}).
\]

\section{The simulation of one running step}\label{s5}

So far, we have only associated the previously constructed formulae with some program 
components or processes. However, it cannot be argued that these formulae model something, 
i.e., they will not always become true, when the described events occur in reality.

Let $C(t)$ and $C(t\!+\!1)$ be some adjacent configurations, we definite 
\[
\Omega^{(0)}(X,P)(\widehat{y}_{t},\widehat{y}_{t+1})
\rightleftharpoons \ [\Psi C(t)(\widehat{y}_{t}) \ \& \
\Phi^{(0)}(P)(\widehat{y}_{t},\widehat{y}_{t+1})] \ 
\rightarrow \: \Psi C(t\!+\!1)(\widehat{y}_{t+1}).
\]
We will prove that the sentence $\forall\widehat{y}_t\forall\widehat{y}_{t+1}
\Omega^{(0)}(X,P)(\widehat{y}_{t},\widehat{y}_{t+1})$ is true if and only if the 
machine $P$ transforms the configuration $C(t)$ to $C(t\!+\!1)$ in one step, i.e., this 
formula models the machine actions at the step $t\!+\!1$.

\begin{rem} 
Since the subformula $\Psi K_1(\widehat{X},\widehat{F})$ does not depend on the 
step number, we can focus only on the subformulae $\Psi K_2(t)$ and $\Psi K_2(t+1)$ of 
formulae $\Psi C(t)$ and $\Psi C(t\!+\!1)$ in two next subsections.
\end{rem}

\subsection{The single-valuedness of modeling}\label{s5.1} 
Let \ $C(t\!+\!1)=K_1\cup K_2(t\!+\!1)$ be a configuration that has arisen from a 
configuration $C(t)=K_1\cup K_2(t)$ as a result of the machine $P$ action at the step 
$t\!+\!1$.

\begin{prop}\label{Pr2}
(i) There exist special values of variables $\widehat{y}_{t}$ such that the 
formula \ $\Psi C(t)(\widehat{y}_{t})$ \ is true, and the truth of \  
$\Phi^{(0)}(P)(\widehat{y}_{t},\widehat{y}_{t+1})$ \ follows from the truth of \ 
$\Psi C(t\!+\!1)(\widehat{y}_{t+1})$ for every tuple $\widehat{y}_{t+1}$.

(ii) If a sentence $\forall\widehat{y}_t\forall\widehat{y}_{t+1}\Omega^{(0)}(X,P)(\widehat{y}
_{t},\widehat{y}_{t+1})$ is true, then the machine $P$ cannot 
convert the configuration $C(t)$ into the configuration, which differs from 
$C(t\!+\!1)$, at the step $t+\!1$. 
\end{prop}

\begin{proof} 
We will prove these assertions simultaneously. Namely, we will select the 
values of the variables $\widehat{y}_{t}$ and $\widehat{y}_{t+1}$ such that a formula 
\[
\Upsilon_{t+1}(\widehat{y}_{t},\widehat{y}_{t+1})\! \rightleftharpoons\!
[\Psi K_2(t)(\widehat{y}_{t})\,\&\,\Phi^{(0)}(P)(\widehat{y}_{t},\widehat{y}_{t
+1})]\!\rightarrow\!\Psi L_2(t\!+\!1)(\widehat{y}_{t+1})
\] 
will be false if the configuration on the second tape $L_2(t\!+\!1)$ differs from the 
real $K_2(t+\!1)$. This implies Item (ii) of the proposition. However, in the beginning, we 
will select the special values of the variables of the tuple $\widehat{y}_{t}$. After 
that when we pick out the values of the corresponding variables of the color 
$t+\!1$, the formulae $\Psi K_2(t\!+\!1)(\widehat{y}_{t+1})$ and \ 
$\Phi^{(0)}(P)(\widehat{y}_{t},\widehat{y}_{t+1})$ will become true or false at 
the same time depending on the values of the variables $\widehat{y}_{t+1}$. 

Let \ $M(k)\!=\!q_{i}\alpha_1,\alpha_2\!\rightarrow\ldots$ be an instruction that is
applicable to the configuration $C(t)$; $(\widehat{\zeta}_1)$ and $(\widehat{\zeta}_2)$ be 
the numbers of the scanned squares. We define \ $\widehat{q}_{t}\!=\!(i)_{2}, \ \widehat{D}
_{t}\!=\!c\alpha_1,\ \widehat{d}_t\!=\!c\alpha_2, \ \widehat{Z}_t\!=\!\widehat{\zeta_1}, \ 
\widehat{z}_{t}\!=\!\widehat{\zeta_2}$. Then $\pi$-formula $\pi_{t}((i)_{2},\alpha_1,\alpha_2,
\widehat{\zeta}_1,\widehat{\zeta}_2)$, which is in the record of the formula 
$\Psi C(t)(\widehat{y}_{t}))$, becomes true. 

Now we consider a formula $\varphi(l)$ that conforms to some instruction 
$M(l)\!=\!q_{b}\theta_1,\theta_2\!\rightarrow\ldots$ that differs from $M(k)$. 
This formula has a timer \ $\pi_{t}((b)_{2},\theta_1,\theta_2,$ $\widehat{U},\widehat{u})$ as 
the first premise. For the selected values of the variables $\widehat{q}_{t}$, 
$\widehat{D}_t$, and $\widehat{d}_{t}$, this timer takes 
the form of \ $(i)_{2}\!\equiv\!(b)_{2}\wedge c\alpha_1\!\equiv\!c\theta_1\wedge c\alpha_2\!
\equiv\!c\theta_2\wedge\ldots$. It is obvious that if \ $i\!\neq b$, \ or 
$\alpha_1\!\neq\!\theta_1$, or \ $\alpha_2\neq\theta_2$, then this $\pi$-formula will be 
false, and the whole $\varphi(l)$ will be true.

Thus, let $\varphi(k)$ be a formula corresponding to the instruction \ 
$M(k)\!=\!q_{i}\alpha_1,\alpha_2\!\rightarrow\!q_{j}\beta_1,\gamma,\beta_2$. The 
quantifier-free part of this formula can become false only for \ $\widehat{U}\!=\!
\widehat{\zeta}_1$ and $\widehat{u}\!=\!\widehat{\zeta}_2$ since the values of variables 
$\widehat{Z}_t$ and $\widehat{z}_t$ are such. Therefore we will consider the case when \ 
$\widehat{U}\!=\!\widehat{\zeta}_1$ and $\widehat{u}\!=\!\widehat{\zeta}_2$. For the same 
reason, we pick out $\widehat{V}\!=\!\widehat{\eta}_1$ and $\widehat{v}\!=\!\widehat{\eta}_2$, 
where the tuple $\widehat{\eta}_l$ satisfies the condition $\kappa_l(\beta_l)(\widehat{\zeta}
_l,\widehat{\eta}_l)$  for $l\!=\!1,2$ (see the formulae \ref{LR1} and \ref{LR2} in Subsection 
\ref{s4.2}). So, the number of the cell, which a head will scann after step $t+\!1$ 
on the $l$th tape, equals \ $(\widehat{\zeta}_l(\beta_l))\!=\!(\widehat{\zeta}_l\!\oplus\!
\widehat{\eta}_l)$.  

Let us assign \ $\widehat{x}_{t}\!=\!\widehat{\zeta}_2(\beta_2)$ and $\widehat{X}\!=\!
\widehat{\zeta}_1(\beta_1)$. Since $\widehat{u}(\beta_2)\!=\!\widehat{\zeta}_2(\beta_2)$, the 
quasi-equation, of the color $t$, which enters in $\langle\Delta^{cop}(\widehat{u}) 
\rangle$ (i.e., in quanti\-fier-free part of $\Delta^{cop}(\widehat{u}))$, is true 
for all $\widehat{w}\!\neq\!\widehat{\zeta}_2(\beta_2)$ irrespective of the values of the 
tuples $\widehat{f}_{t}$ and $\widehat{g}$ according to Lemma \ref{trcl}(i,ii); while when 
$\widehat{w}\!=\!\widehat{\zeta}_2$, the premise of $\langle\Delta^{cop}(\widehat{u})
\rangle$ is false. For the same reason, all clauses that are included in $\Psi K_2(t)$ and 
$\Psi K_1$ are true, except the clauses \ $\psi_{2,t}(\widehat{\zeta}_2(\beta_2)\rightarrow\!
\lambda_2)$ and $\psi_{1}(\widehat{\zeta}_1(\beta_1)\rightarrow\!\lambda_1)$, where 
$\lambda_l$ is that symbol, which a head will view after step $t+\!1$ on the $l$th tape for 
$l\!=\!1,2$. We set the values of two tuples $\widehat{F}$ and $\widehat{f}_{t}$ as 
$c\lambda_1$ and $c\lambda_2$, respectively. Now, the questionable clauses from $\Psi C(t)$ 
become true because their premises and conclusions are such. All clauses of color $t$ from 
$\langle\varphi(k)\rangle$ become also true.

So, we have selected the values of the variables $\widehat{y}_t$, i.e., of all basic ones 
of color $t$. It remains only to assign the values of the variables of color $t\!+\!1$.

We recall that $\lambda_l$ is the symbol that will be in the $(\widehat{\zeta}_l(\beta_l))$th 
cell on $l$th tape for $l\!=\!1,2$ at the instant $t\!+\!1$. We 
define \ $\widehat{D}_{t+1}\!=\!c\lambda_1$, \ $\widehat{d}_{t+1}\!=\!c\lambda_2$, \ 
$\widehat{Z}_{t+1}\!=\!\widehat{\zeta}_1(\beta_1)\!=\!\widehat{\zeta}_1\oplus\widehat{\eta}
_1$,  \ $\widehat{z}_{t+1}\!=\!\widehat{\zeta}_2(\beta_2)\!=\!\widehat{\zeta}_2\oplus
\widehat{\eta}_2$, \ and \ $\widehat{q}_{t+1}\!=\!(j)_{2}$.

With this assignment and the chosen values of variables $\widehat{U}$, $\widehat{u}$, 
$\widehat{V}$, and $\widehat{v}$, the timer $\pi_{t+1}((j)_2,\lambda_1,\lambda_2,
\widehat{\zeta}_1(\beta_1),\widehat{\zeta}_2(\beta_2))$, which enters into the record of $\Psi 
K_2(t+\!1)$, becomes true; while the $\pi$-formula contained in $\Psi L_2(t+\!1)$ is false if 
it differs from "proper". But the conclusion of the quantifier-free part 
$\langle\varphi(k)\rangle$ of the formula $\varphi(k)$ contains a slightly different timer 
$\pi_{t+1}((j)_{2},\widehat{H},\widehat{h},\widehat{U}(\beta_1),$ $\widehat{u}(\beta)_2)$. 
In this timer, the equivalencies \ $\widehat{D}_{t+1}\!\equiv\!\widehat{H}$ \ and \ 
$\widehat{d}_{t+1}\!\equiv\!\widehat{h}$ \ raise doubts for the time being.

If \ $\widehat{H}\!\neq\!c\lambda_1$ \ or \ $\widehat{h}\!\neq\!c\lambda_2$, then the  
formula  $\Gamma^{ret}(\beta_1,\beta_2)$ will be false for the assidned values of 
$\widehat{X}$, $\widehat{x}_t$, $\widehat{F}$, and $\widehat{f}_t$ and the calculated 
$\widehat{\zeta}_1(\beta_1)$ and $\widehat{\zeta}_2(\beta_2)$. 
Hence, the whole formula $\langle\varphi(k)\rangle$ will be true in this case because its 
third premise $\Gamma^{ret}(\beta_1,\beta_2)$ is false. When $\widehat{H}\!=\!c\lambda_1$ 
and $\widehat{h}\!=\!c\lambda_2$, the terminal $\pi$-formula in $\langle\varphi(k)\rangle$ 
becomes true.  

When the "incorrect" formula $\Psi L_2(t+\!1)$ has a mistake in the record of the clause 
$\psi_{2,t+1}(\widehat{\zeta}_2\rightarrow\!\gamma)$ (see the definition of the formula 
$\Delta^{cop}(\widehat{u})$ in Subsection \ref{s4.2}), we will assign 
$\widehat{x}_{t+1}\!=\!\widehat{\zeta}_2$ \ and \ $\widehat{f}_{t+1}\!=\!c\,\gamma$. If this 
fragment is as it should be, but there is another "incorrect" clause 
$\psi_{t+1}(\widehat{\mu}\!\rightarrow\!\rho)$, where $\rho$ is different from "real" $\delta
$, we will define \ $\widehat{x}_{2,t+1}\!=\!\widehat{\mu}$ and \ $\widehat{f}_{t+1}\!=\!
\widehat{g}\!=\!c\delta$ (we note that \ 1) this is the only case when we need to set the 
values of the variables $\widehat{g}$; \ 2) in this case, we essentially use the formula 
$\Delta^{cop}(\widehat{u})$ in the form of disjunction --- see Remark \ref{remQu}). 
Now, the quasi-equations of the color $t\!+\!1$ in the formulae $\langle\Delta^{cop}
(\widehat{u})\rangle$ and $\Delta^{wr}(\widehat{u})$ are valid in both of 
these cases on the grounds of Lemma \ref{trcl}(i,ii) or because their premises and conclusions 
are true. Therefore, the whole formula $\langle\varphi(k)\rangle$ is true. All the clauses 
contained in $\Psi K_2(t\!+\!1)$ are true for the same reasons.

We obtain as a result that any formula $\varphi(l)$ is true for the above-selected  
values of the primary variables, so the entire conjunction 
$\Phi^{(0)}(P)$ is true. Since the premise and conclusion of the \
$\Omega^{(0)}(X,P)(\widehat{y}_{t},\widehat{y}_{t+1})$ are true, and the
configurations $K_2(t+\!1)$ and $L_2(t+\!1)$ are different, the "incorrect" formula 
$\Upsilon_{t+1}$ is false because of the choice of the values of basic variables. 

As the configuration $L_2(t\!+\!1)$ may differ from the real 
$K_2(t+\!1)$ in any place, Item (i) is established too. 
\end{proof} 

\subsection{The sufficiency of modeling}\label{s5.2} 
We will now prove a converse to Proposition \ref{Pr2}(ii).

\begin{prop}\label{Pr3} 
Let \ $C(t\!+\!1)=K_1\cup K_2(t\!+\!1)$ be a configuration that 
has arisen from a configuration $C(t)=K_1\cup K_2(t)$ as a result of an action of the machine 
$P$ at the step $t\!+\!1$. Then the formula
$\Omega^{(0)}(X,P)(\widehat{y}_{t},\widehat{y}_{t+1})$ is true for all the values variables
$\widehat{y}_{t}$ and $\widehat{y}_{t+1}$. 
\end{prop}
 
\begin{proof} 
Let \ $M(k)\!=\!q_{i}\alpha_1,\alpha_2\!\rightarrow\!q_{j},\beta_1,\gamma,\beta_2$ 
be the instruction that transforms the configuration $C(t)$ into the $C(t\!+\!1)$; 
and $\varphi(k)(\widehat{y}_{t},\widehat{y}_{t+1})$ be a formula, which is 
written for this instruction. This formula is a consequence of \ 
$\Phi^{(0)}(P)(\widehat{y}_{t},\widehat{y}_{t+1})$.

Let us replace $\varphi(k)$ by a conjunction of formulae
$\varphi(k)(\widehat{\mu}_1,\widehat{\mu}_2)$; they are each obtained as the result of the
substitution of the various values of the universal variables $\widehat{U}$ and $\widehat{u}$ 
for the variables themselves. Every formula $\varphi(k)(\widehat{\mu}_1,\mu_2)$ contains the 
subformula \ $\widehat{q_{t}}\!\equiv\!(i)_{2}\wedge\widehat{D}_t\!\equiv\!c\alpha_1\wedge
\widehat{d}_{t}\!\equiv\!c\alpha_2\wedge\widehat{Z}_t\!\equiv\!\widehat{\mu}_1\wedge
\widehat{z}_{t}\!\equiv\!\widehat{\mu}_2$ \ as the first premise. One of these subformulae 
coincides with the only timer \ $\pi_{t}((i)_2,\alpha_1,\alpha_2,\widehat{\zeta}_2,
\widehat{\zeta}_2)$ included in $\Psi K_2(t)$ for \ $\widehat{U}\!=\!\widehat{\mu}_1\!=\!
\widehat{\zeta}_1$ \ and \ $\widehat{u}\!=\!\widehat{\mu}_2\!=\!\widehat{\zeta}_2$,\ as the 
instruction $M(k)$ is applicable to the configuration $C(t)$. 

For these \ $\widehat{U}\!=\!\!\widehat{\zeta}_1$ and \ $\widehat{u}\!=\!\widehat{\zeta}_2$, 
there are tuples $\widehat{\eta}_1$ and $\widehat{\eta}_2$ \ such that the conditions 
$\kappa_1(\beta_1)(\widehat{U},\widehat{V})$ and $\kappa_2(\beta_2)(\widehat{u},\widehat{v})$ 
are valid for $\widehat{V}\!=\!\widehat{\eta}_1$ and $\widehat{v}\!=\!\widehat{\eta}_2$. So, 
the second premise of $\varphi(k)$ is true for suitable values of variables $\widehat{U}$, 
$\widehat{u}$, $\widehat{V}$, and $\widehat{v}$.

Using these $\widehat{\eta}_1$ and $\widehat{\eta}_2$, one can find $\widehat{U}(\beta_1)$ and 
$\widehat{u}(\beta_2)$ correspondently as $\widehat{\zeta}_1\!\oplus\widehat{\eta}_1$ and 
$\widehat{\zeta}_2\!\oplus\widehat{\eta}_2$. We recall that the machine $P$ cannot go beyond 
the left edges of tapes, hence \ $0\!\leqslant\!(\widehat{\zeta}_1(\beta_1))\!\leqslant\!n\!+
\!1$ \ and \ $0\!\leqslant\!(\widehat{\zeta}_2(\beta_2))\!\leqslant\!T\!=\!\exp(m)$ (see the 
second paragraph in Subsection \ref{s4.1}); therefore the formulae $\Psi K_1$ and $\Psi K_2(t)
$ accordingly contain the quasi-equations $\psi_1(\widehat{\zeta}_1(\beta_1)\rightarrow
\lambda_1)$ and $\psi_{2,t}(\widehat{\zeta}_2 (\beta_2)\rightarrow\lambda_2)$ for some symbols 
$\lambda_1$ and $\lambda_2$. The conjunction of these clauses coincides with $\Gamma^{ret}
(\beta_1,\beta_2)$, when $\beta_1\!\in\!\{L,R,S\}$ and $\beta_2\!\in\!\{L,R\}$ for 
$\widehat{H}\!=\!c\lambda_1$ and $\widehat{h}\!=\!c\lambda_2$; while if $\beta_2\!=\!S$, then 
$\Gamma^{ret}_2(\beta_2)$ becomes true equivance for $\widehat{h}\!=\!c\,\gamma$. 

So, each of the three premises in the formula $\varphi(k)$ coincides with the appropriate 
subformula of $\Psi C(t)$ or is true for some values of universal auxiliary variables from
the list at the beginning of $\varphi(k)$. It follows from this that 
\[
\Psi K_2(t) \ \ \& \ \ \Delta^{cop}(\widehat{\zeta}_2) \ \ \& \ \ \Delta^{wr}
(\widehat{\zeta}_2) \ \ \& \ \ \pi_{t+1}((j)_2,\lambda_1,\lambda_2,\widehat{\zeta}
_1(\beta_1),\widehat{\zeta}_2(\beta_2))
\] 
is the consequence of the formulae $\Psi K_2(t)$ and $\Phi^{(0)}(P)(\widehat{y}_{t},
\widehat{y}_{t+1})$.

The formula $\Delta^{cop}(\widehat{\zeta}_2)$ begins with the quantifiers 
$\forall\,\widehat{w}$. Let us replace this formula with a conjunction that is equivalent  
to it; we substitute all possible values of variables $\widehat{w}$ into this formula to  
this effect. For every value of $\widehat{w}$, there is a 
unique value of the tuple $\widehat{g}$ such that the clause \ 
$\psi_{2,t}(\widehat{w}\!\rightarrow\!\widehat{g})$ enters into the formula 
$\Psi K_2(t)$. When these values of $\widehat{g}$ are substituted in their 
places, we will obtain all the quasi-equations from $\Psi K_2(t\!+\!1)$, except one, which 
coincides with \ $\Delta^{wr}(\widehat{\zeta}_2)$. 
\end{proof}

\section{End of the proof of the main theorem}\label{s6}

\subsection{The simulation of the exponential computations}\label{s6.1}
  
We define the formulae that model $e(k)\!\rightleftharpoons\!\exp(k)$
running steps of a machine $P$ when it applies to a configuration $C(t)$:
\[
\Omega^{(k)}(X,P)(\widehat{y}_{t},\widehat{y}_{t+e(k)}) \ 
\rightleftharpoons\ [\Psi C(t)(\widehat{y}_{t}) \ \& \ \Phi^{(k)}(P)
(\widehat{y}_{t},\widehat{y}_{t+e(k)})]
\rightarrow \ \Psi C(t\!+\!e(k))(\widehat{y}_{t+e(k)}).
\]

\begin{prop}\label{Pr4} 
Let $t,k\!\geqslant\!0$ be the integers such that \
$t\!+\!e(k)\!\leqslant\!T\!=\!\exp(m)$.

(i) If the machine $P$ transforms the configuration $C(t)$ to $C(t\!+\!e(k))$ 
within $e(k)$ steps, then there are special values of 
variables \ $\widehat{y}_{t}$ such that the formula \ 
$\Psi C(t)(\widehat{y}_{t})$ \ is true, and for all \ $\widehat{y}_{t+e(k)}$, 
whenever the formula $\Psi C(t\!+\!e(k))(\widehat{y}_{t+e(k)})$ is true, \  
$\Phi^{(k)}(P)(\widehat{y}_{t},\widehat{y}_{t+e(k)})$ is also true.

(ii) The formula \ $\Omega^{(k)}(X,P)(\widehat{y}_{t},\widehat{y}_{t+e(k)})$ 
is valid for all values of variables \ $\widehat{y}_{t},\widehat{y}_{t+e(k)}$ 
(i.e., identically true) if and only if the machine $P$ 
converts the configuration $C(t)$ into $C(t\!+\!e(k))$ within $e(k)$ steps. 
\end{prop}

\begin{proof}
Induction on the parameter $k$. For $k\!=\!0$, Item (i) is Proposition \ref{Pr2}(i), 
and Item (ii) follows from Propositions \ref{Pr2}(ii) and \ref{Pr3}.

We start the proof of the inductive step by rewriting the formula \break
$\Phi^{(k+1)}(P)(\widehat{y}_{t},\widehat{y}_{t+e(k+1)})$ in the equivalent, 
but longer form:
\begin{multline*}
\exists\,\widehat{Y}\bigl\{\forall\,\widehat{a}\,\forall\,
\widehat{b} \bigl[(\widehat{y}_{t}\!\equiv\!\widehat{a} \ \wedge \
\widehat{Y}\!\equiv\!\widehat{b}) \rightarrow \Phi^{(k)}(P)(\widehat{a},
\widehat{b}) \bigr] \quad \& \\
\& \quad \forall\,\widehat{a}\,\forall\,\widehat{b}
\bigl[(\widehat{Y}\!\equiv\!\widehat{a} \ \wedge \
\widehat{b}\!\equiv\!\widehat{y}_{t+e(k+1)}) \rightarrow 
\Phi^{(k)}(P)(\widehat{a},\widehat{b})\bigr]\bigr\}.
\end{multline*}   
The following formula results from this immediately:
\[
\Xi_{k+1} \ \rightleftharpoons \ \exists\,\widehat{Y}\bigl\{\Phi^{(k)}(P)
(\widehat{y}_{t},\widehat{Y})
\ \& \ \Phi^{(k)}(P)(\widehat{Y},\widehat{y}_{t+e(r+1)})\bigr\}.
\]
On the other hand, each of the two implications which are included in the long form of 
the formula \ $\Phi^{(k+1)}(P)(\widehat{y}_{t},\widehat{y}_{t+e(k+1)})$ can be false 
only when the equivalences existing in its premise are valid. Hence, this formula 
is tantamount to $\Xi_{k+1}$. 

Let the machine $P$ transform the configuration $C(t)$ into $C(t+\!e(k))$ within 
$e(k)$ steps and convert the latter into $C(t\!+\!e(k\!+\!1))$ within the same time.

By the inductive hypothesis of Item (ii) (we recall that the induction is  
carried out over a single parameter $k$), the formula \ $\Omega^{(k)}(P)
(\widehat{y}_{t},\widehat{y}_{t+e(k)})$ is identically true for any $t$ such
that \ $t+\!e(k)\!\leqslant\!T$, and hence, it is identically true for an 
arbitrarily chosen $t$ and $t_1\!=\!t\!+\!e(k)$ provided that \
$t\!+\!e(k\!+\!1)\!=\!t_1\!+\!e(k)\!\leqslant\!T$. Thus, the formulae
\[ 
[\Psi C(t)(\widehat{y}_{t}) \
\& \ \Phi^{(k)}(P)(\widehat{y}_{t},\widehat{y}_{t+\!e(k)})] \ \rightarrow \ \Psi 
C(t\!+\!e(k))(\widehat{y}_{t+e(k)})
\]
and 
\[
[\Psi C(t\!+\!e(k))(\widehat{y}_{t+e(k)}) \ \& \ \Phi^{(r)}(P)(\widehat{y}_{t+e(k)},
\widehat{y}_{t+e(k+1)})] \ \rightarrow
 \Psi C(t\!+\!e(k\!+\!1))(\widehat{y}_{t+e(k+1)})
\]   
are identically true. Therefore, when we change the variables  $\widehat{y}_{t+e(k)}$ 
under the sign of the quantifier, we obtain from this that the following formula  \ 
\[
\forall\widehat{Y}\{[\Psi C(t)(\widehat{y}_{t}) \ \& \ \Phi^{(k)}(P)
(\widehat{y}_{t},\widehat{Y}) \ \& \ \Phi^{(k)}(P)(\widehat{Y},
\widehat{y}_{t+e(r+1)})]\rightarrow  
 \ \Psi C(t\!+\!e(k\!+\!1))(\widehat{y}_{t+e(k+1)})\}
\]   
is also identically true. This formula is equivalent to \
$$[(\Psi C(t)(\widehat{y}_{t}) \ \& \ \Xi_{k+1}) \ \rightarrow \ \Psi C(t\!+\!e(k\!+\!1))]
(\widehat{y}_{t},\widehat{y}_{t+e(k+1)})$$
because the universal quantifiers will be interchanged with the quantifiers of existence when 
they are introduced into the premise of the implication. Since the premise of the 
formula  $\Omega^{(k+1)}(X,P)(\widehat{y}_{t},\widehat{y}_{t+e(k+1)})$  is equivalent to 
$\Psi C(t) \ \& \ \Xi_{k+1}$ under the preceding argument, the inductive   
step of Item~(ii) is proven in one direction.

Now, let the configurations \ $L(t+\!e(k+\!1))$ and \ $C(t+\!e(k+\!1))$ be 
different. For some values \ $\widehat{Y}_{1}$ and \ $\widehat{Y}_{2}$ of variables \ 
$\widehat{y}_{t+e(k)}$ and \ $\widehat{y}_{t+e(k+1)}$, respectively,  the formula
$\{[\Psi C(t\!+\!e(k)) \ \& \ \Phi^{(k)}(P)]\rightarrow\Psi L(t+e(k+\!1))\}(\widehat{Y}_{1},
\widehat{Y}_{2})$ \ is false by the inductive assumption of Items (ii). Therefore, its 
conclusion \ $\Psi L(t+e(k+\!1))(\widehat{Y}_{2})$ \ is false, 
but both its premises \ $\Phi^{(k)}(P)(\widehat{Y}_{1},\widehat{Y}_{2})$ \ and \ 
$\Psi C(t+e(k))(\widehat{Y}_{1})$ are true. Since the last formula and 
$\Psi C(t)(\widehat{Y}_{0})$ are true for some special $\widehat{Y}_{0}$, 
which exists due to the induction proposition of Item~(i), the formula \
$\Phi^{(k)}(P)(\widehat{Y}_{0},\widehat{Y}_{1})$ is true. So, the implication \ 
\[
\{[\Psi C(t) \ \& \ \Phi^{(k+1)}(P)]\rightarrow\Psi
L(t\!+\!e(k\!+\!1))\}(\widehat{Y}_{0},\widehat{Y}_{2})
\]
has a true premise, and a false conclusion, therefore it is not 
identically true. Item~(ii) is proven.

Inasmuch as the configuration $L(t\!+\!e(k+\!1))$ may differ from the proper 
one at any position, to finish the proof of Item~(i), we set the values $\widehat{Y}_{1}$ of 
the variables $\widehat{y}_{t+e(k)}$ in a specific manner, using the inductive 
hypothesis of Item (ii) and then Item (i) for the configurations $C(t+\!e(k))$ and 
$C(t+\!e(k+\!1))$, as $P$ transforms the first of them to the second.
\end{proof}
 
\subsection{The short recording of the initial configuration and the condition of the successful 
run termination} \label{s6.2}  

We cannot fully describe the final configuration on the second tape when the machine arrives at 
one of the two states $q_{acc}$ or $q_{rej}$. Firstly, this final configuration may be very long. 
Secondly, and most importantly, it is unknown to us --- we only know the input string $X$ on the 
first tape and the machine's program $P$. However, we do not need the entire final configuration.

Since we have the instructions for the machine's run at idle (see the beginning of 
Section~\ref{s4}), the statement that the machine $P$ accepts an input string $X$ within \ 
$T\!=\!\exp(m(n))$ steps can be written rather briefly --- utilizing one quantifier-free 
formula of the color $T$: \ $\chi(\omega)\rightleftharpoons\widehat{q}\,_{T}\!\equiv\!(1)_2$ 
because the accepting state has the number one. This formula is a full 
analog of subformula $I$ of Cook's formula from \cite{Coo71}, and it has 
$\mathcal{O}(r)$ symbols in length nonmetering the indices. The 
writing of the first index $T$ (in the tuple $\widehat{q}\,_{T}$) occupies $m\!+\!1\!=\!m(n)\!
+\!1$ bits. The maximum length of the second indices is $\mathcal{O}\lceil\log r\rceil$, 
so we have \ $|\chi(\omega)|\!=\!\mathcal{O}(r\cdot(m(n)\!+\!\lceil\log r\rceil))$.

The formula $\Psi C(t)(\widehat{y}_t)$ was introduced in Subsection \ref{s4.3} to describe 
a confi\-gu\-ration arising after step $t$. It is a very lengthy:
$|\Psi C(t)|>m\cdot\exp(m)$. However, the initial configuration is simple enough and allows us 
to give a relatively short description. We have already described the entries on the first 
tape in Subsection \ref{s4.3.3} since they remain unchanged throughout the work. This is done 
using the following formula $\chi_1(X)$ with the free variables $\widehat{X}\!=\!\langle X_0,
\ldots,X_m\rangle$ and $\widehat{F}\!=\!\langle F_0,\ldots,F_r\rangle$: 
\begin{eqnarray*} 
\chi_1(X)\rightleftharpoons\Psi K_1(\widehat{X},\widehat{F}) \quad \rightleftharpoons \quad 
\psi_1(\widehat{0}\rightarrow\rhd) \ \& \ 
\bigwedge_{1\leqslant\,(\widehat{\xi})\leqslant\,n}
\psi_{1}(\widehat{\xi}\!\rightarrow\!\chi(\widehat{\xi})) \quad \& \\
\& \quad \forall\widehat{u}_0\:\bigl(\widehat{u}_0>(n)_2 \ \rightarrow \ \psi_1(\widehat{u}
_0\rightarrow\Lambda)\bigr).
\end{eqnarray*}
Let us pay attention that the variable $X$ denotes the input string, but not the tuple of the 
variables $\widehat{X}$. The sense of $\chi_1(X)$ is as follows: the $\rhd$ symbol is located 
in the zeroth cell, the input string $X=\langle\chi_1,\ldots,\chi_n\rangle$ is written next to 
$\rhd$ in the squares $1,2,\ldots,n$, and the remaining cells of the tape are empty.  
  
The entries on the second tape before starting work are described even more simply since all 
its cells are empty, except for the zeroth one, which contains the $\rhd$ symbol:
\[
\chi_2(0)\rightleftharpoons\psi_{2,0}((0)_2\!\rightarrow\!\rhd) \ \& \ \forall\widehat{u}
_0[\widehat{u}_0\geqslant(1)_2\rightarrow\psi_{2,0}(\widehat{u}_0\rightarrow\Lambda)].
\]  

It remains to record that the machine is in the start state $q_0$; both of its heads are 
pointed at zeroth cells and see the $\rhd$ symbols there. The single timer of color 0 describes 
all this \ \quad $\pi_0((0)_2,\rhd,\rhd,(0)_2,(0)_2)$.

The whole initial configuration is described by the formula
\[
\chi(0)(\widehat{y}_0) \ \rightleftharpoons \ \chi_1(X) \quad \& \quad \chi_2(0) \quad \& \quad 
\pi_0((0)_2,\rhd,\rhd,(0)_2,(0)_2),
\]
It is easy see that the formula $\chi(0)(\widehat{y}_0)$ is the brief form for 
$\Psi C(t)(\widehat{y}_0)$ introduced in Subsection \ref{s4.3.3} for $t\!=\!0$. 

We recall that $r$ positions are reserved for recording the codes of the symbols of the 
working alphabet of simulated machines and the indication of the machine state's numbers in 
Subsection \ref{s4.1}; and if $|X|\!=\!n$, then $|\widehat{X}|\!=\!\mathcal{O}(\lceil\log n
\rceil\cdot\lceil\log\lceil\log n\rceil\rceil)$.

\begin{lem}\label{lchi0} 
$|\chi(0)(\widehat{y}_{0})|\!\leqslant\!D_1\cdot n\cdot\lceil
\log n\rceil\cdot\lceil\log\lceil\log n\rceil\rceil\!+\!D_2\cdot m(n)\cdot\lceil\log m(n)
\rceil$ for the proper constants $D_{1}$, $D_2$ and \ $|X|\!=\!n\!>\!r$. 
\end{lem}

\begin{proof} 
This follows from Lemmata \ref{ll2} and \ref{lcl} by immediate calculation.
\end{proof}

\subsection{The simulating formula $\Omega^{(m)}(X,P)$}\label{s6.3} 

Let us define
\begin{eqnarray}\label{6.1}
\Omega^{(m)}(X,P) \ \rightleftharpoons
\ \forall\widehat{y}_{0},\widehat{y}\,_{T} \ \Bigl\{\Bigl[
\ \chi(0)(\widehat{y}_{0}) \ \& \ \exists\,\widehat{Y}_{m}\forall\,\widehat{a}_{m}\forall\,
\widehat{b}_{m}\ldots\exists\,\widehat{Y}_1\forall\,\widehat{a}_{1}\forall\,
\widehat{b}_{1} \nonumber  \\
\bigl\{\bigwedge_{1\leqslant\,k\leqslant\,m}\bigl[
(\widehat{a}_{k+1}\!\equiv\!\widehat{a}_{k}\wedge
\widehat{Y}_{k}\!\equiv\!\widehat{b}_{k}) \ \vee
(\widehat{Y}_{k}\!\equiv\!\widehat{a}_{k}\wedge\widehat{b}_{k}\!
\equiv\!\widehat{b}_{k+1})\bigr] \rightarrow  
\Phi^{(0)}(P)(\widehat{a}_{1},\widehat{b}_{1})\bigr\}\Bigr] \rightarrow\ \\
\rightarrow \ \quad \chi(\omega)(\widehat{y}\,_{T})\Bigr\},\nonumber
\end{eqnarray}
here we designated $\widehat{a}_{m+1}\!=\!\widehat{y}_{0}, \
\widehat{b}_{m+1}\!=\!\widehat{y}\,_{T}$ in the record of the "big" 
conjunction for the sake of brevity. 

\begin{prop}\label{Pr5} 
The formula $\Omega^{(m)}(X,P)$ has the property (i) from 
the statement of Teorema \ref{Main}. In other words, this sentence is true  
if and only if the machine $P$ accepts the input $X$ within \ $T=\exp(m(n))$ steps. 
\end{prop}

\begin{proof} 
Let $\Theta_k\!=\!\Theta_k(\widehat{a}_{k},\widehat{a}_{k+1},\widehat{Y}_k,
\widehat{b}_{k},\widehat{b}_{k+1})$ be the denotation for the expression \ 
$(\widehat{a}_{k+1}\!\equiv\!\widehat{a}_{k}\wedge\widehat{Y}_{k}\!\equiv\!\widehat{b}_{k})
\vee(\widehat{Y}_{k}\!\equiv\!\widehat{a}_{k}\wedge\widehat{b}_{k}\!\equiv\!\widehat{b}_{k+1})
$. The part of the formula $\Omega^{(m)}(X,P)$, which is located in the big square brackets 
in Equation \ref{6.1}, almost coincides with the first of the two following formulae and is  
equally matched to the second (according to the agreement of Subsection \ref{s3.2} that 
conjunction connects more intimately than an implication):
\begin{eqnarray*}
1) \quad \chi(0)(\widehat{y}_{0}) \ \& \
\exists\,\widehat{Y}_{m}\forall\,\widehat{a}_{m}\forall\,\widehat{b}_{m}
\ldots\exists\,\widehat{Y}_1\forall\,\widehat{a}_{1}\forall\,\widehat{b}_{1} \
\bigl\{\bigwedge\limits_{1\leqslant\,k\leqslant\,m}\Theta_{k}\rightarrow
\Phi^{(0)}(P)(\widehat{a}_{1},\widehat{b}_{1})\bigr\};  \\
2) \ \ \chi(0)(\widehat{y}_{0}) \ \& \ \exists\,\widehat{Y}_{m}\forall
\,\widehat{a}_{m}\forall\,\widehat{b}_{m}\ldots\exists\,\widehat{Y}_1\forall\,
\widehat{a}_{1}\forall\,\widehat{b}_{1}
\Bigl(\Theta_{m}\rightarrow\bigl(\Theta_{m-1}\rightarrow(\ldots\rightarrow(\Theta_{1}\to\\
\rightarrow \ \Phi^{(0)}(P)(\widehat{a}_{1},\widehat{b}_{1}))\ldots)\bigr) \Bigr).
\end{eqnarray*}
 If we carry the quantifiers through the subformulae, which do not contain 
the corresponding variables, then we will conclude that the second formula is tantamount to 
the third one:
\begin{eqnarray*}
3) \ \ \chi(0)(\widehat{y}_{0}) \ \& \ \exists\,\widehat{Y}_{m}
\forall\,\widehat{a}_{m}\forall\,\widehat{b}_{m}\Bigl(\Theta_{m} \ \rightarrow 
 \ \exists\,\widehat{Y}_{m-1}\forall\,\widehat{a}_{m-1}\forall\,
\widehat{b}_{m-1}\bigl(\Theta_{m-1}\rightarrow(\ldots \rightarrow  \\
\rightarrow \ \exists\,\widehat{Y}_1\forall\,\widehat{a}_{1}\forall\,
\widehat{b}_{1}(\Theta_{1}\rightarrow
\Phi^{(0)}(P)(\widehat{a}_{1},\widehat{b}_{1}))\ldots)\bigr)\Bigr).
\end{eqnarray*}

According to the definition, the formula \ $\exists\,\widehat{Y}_k\forall\,\widehat{a}_{k}\forall
\,\widehat{b}_{k}(\Theta_{k}\to\Phi^{(k-1)}(P)(\widehat{a}_{k},\widehat{b}_{k}))$ \ contracts 
into $\Phi^{(k)}(P)(\widehat{a}_{k+1},\widehat{b}_{k+1})$. Therefore the whole sentence 
$\Omega^{(m)}(X,P)$ is equivalent to \ $\forall\widehat{y}_{0},\widehat{y}\,_{T}
\bigl[\bigl(\chi(0) \ \& \ \Phi^{(m)}(P)\bigr)\to\chi(\omega)\bigr]$. Notice that the formula
\[
\forall\widehat{y}_{0},\widehat{y}\,_{T}\bigl[\bigl(\Psi C(0)(\widehat{y}_0) \ \& \ 
\Phi^{(m)}(P)(\widehat{y}_{0},\widehat{y}\,_{T})\bigr)\to\Psi(T)(\widehat{y}_{T})\bigr]
\]
is true by Proposition \ref{Pr4}(ii), and $\chi(\omega)$ is a product term of the formula 
$\Psi(T)(\widehat{y}_{T})$, provided that the machine $P$ accepts an input string $X$ within $T$ 
steps. Consequently, based on the fact that formulae $\chi(0)(\widehat{y}_0)$ and 
$\Psi C(0)(\widehat{y}_0)$ are tantamount, one could say that the formula \ref{6.1} is the 
modeling formula.  
\end{proof}

\subsection{The time of writing of $\Omega^{(m)}(X,P)$ and its length}\label{s6.4} 

Let us discuss the possibility of writing this formula for the various functions $m(n)$.

\begin{lem}\label{sp-con} 
Let $g(n)$ be a fully space-constructible function. Then this 
function is also indexed-constructible, i.e., for each input string having length $n$ on the 
first tape, one can fill exactly $g(n)\!+\!1$ square on the second tape and then number the 
filled cells in the binary representation using \ $\mathcal{O}(g(n)\!\cdot\!\lceil\log g(n)
\rceil)$ memory cells within time bounded the value of some polynomial on $\max\{n,g(n)\}$, 
provided that the working alphabet has at least five symbols: $\rhd, \Lambda, 0, 1$, and one 
more, for instance, $X$.
\end{lem}

\begin{proof} 
At first, we fill exactly $g(n)\!+\!1$ squares on the second tape with the $X$ symbol. One can 
make this using exactly $g(n)\!+\!1$ memory cells since the function $g(n)$ is fully 
space-constructible. Then, we count the filled cells beginning with the zero number and apply 
a reverse binary representation, \eg the number $6$ will be written down as 
$011$, but not $110$. For this purpose, we write the tuple $Y$ consisting of 0 and 1 to the right 
of $(g(n)\!+\!1)$-tuple $(X,\ldots,X)$. The tuple $Y$ is $\lceil\log(g(n)\!+\!1)\rceil$ 
symbols in length. One can fill these $g(n)\!+\!|Y|\!+\!1$ squares 
utilizing the same amount of memory. Next, we fill the zone $Y$ with zeros. After that, we move 
apart the adjacent symbols $X$ by $|Y|$ positions and fill these gaps with zeros. Possibly, 
we will use a few more cells than $(g(n)\!+\!1)\cdot|Y|$ for this expansion. It remains 
only to number the $X$ symbols, using gaps filled with zeros for this, but now we write the binary 
representation of the numbers in the usual way. If $g(n)\!=\!\lceil\log n\rceil$, then as a 
result, we will obtain the tuple $\widehat{X}$ involved in the record of the simulating formula.

All this is feasible in a polynomial time on $\max\{n,g(n)\}$. 
\end{proof}  

\begin{rems} 
(1) It is quite probable that this lemma can be noticeably strengthened. In 
particular, one can most likely get by with the alphabet $\{\rhd,\Lambda,0,1\}$. 

(2) All functions, interesting for us, are fully space-constructible: $\lceil\log n\rceil,\ n$, 
$cn^d$, and $\exp_k(n)$. The functions \ $n\cdot\lceil\log n\rceil$ \ and \ $n\cdot\lceil\log 
n\rceil\cdot\lceil\log\lceil\log n\rceil\rceil$ are uninteresting for us concerning the property 
"be space-constructible" \ because the summand $D_1\cdot n\cdot\lceil\log n\rceil\cdot\lceil
\log\lceil\log n\rceil\rceil$ arise in the upper estimate of the length of modeling formula 
(see inequalities \ref{in1} and Lemma \ref{lchi0}) due to the fact the input string is 
described by $n$ clauses, each of which is $\mathcal{O}(\lceil\log n\rceil\cdot\lceil\log
\lceil\log n\rceil\rceil)$ in length. 
\end{rems}

We recall that the disjunction of equivalencies \ 
$(\widehat{a}_{k+1}\!\equiv\!\widehat{a}_{k}\wedge
\widehat{Y}_{k}\!\equiv\!\widehat{b}_{k})\vee(\widehat{Y}_{k}\!\equiv\!
\widehat{a}_{k}\wedge\widehat{b}_{k}\!\equiv\!\widehat{b}_{k+1})$ \ entered in the formula 
$\Omega^{(m)}(X,P)$ was denoted as $\Theta_k$ in the proof of Proposition \ref{Pr5}; and the 
lengths of the formulae are calculated in the natural language --- see Subsection \ref{s3.2}.

\begin{lem}\label{lfrOm} 
The length of the subformula \ 
\[
\Upsilon(m) \ \rightleftharpoons \ \exists\,\widehat{Y}_{m}\forall\,\widehat{a}_{m}\forall\,
\widehat{b}_{m}\ldots\exists\,\widehat{Y}_1\forall\,\widehat{a}_{1}\forall\,\widehat{b}_{1} \
\bigl\{\bigwedge\limits_{1\leqslant\,k\leqslant\,m}\Theta_{k}\rightarrow
\Phi^{(0)}(P)(\widehat{a}_{1},\widehat{b}_{1})\bigr\}
\]
is \ $\mathcal{O}((|P|\cdot m(n)+m^2(n))\cdot\lceil\log m(n)\rceil)$ \ 
for sufficiently long $X$ and \ $m\!=\!m(n)\!\geqslant\!\lceil\log n\rceil$. 
\end{lem} 

\begin{proof} 
According to Lemmata \ref{lphi} and \ref{trcl}(iii), 
\[
|\Phi^{(0)}(P)(\widehat{y}_{t},\widehat{y}_{t+1})|\!\leqslant\!
C_{1}\cdot|P|\cdot|\varphi(N)|\!\leqslant\!C_2\cdot|P|\cdot(m+\!r+\!s)\cdot
\max\{\log m,\log r,\log s\}
\]  
for appropriate constants $C_1$ and $C_2$. The tuple $\widehat{y}_t$ consists of $2m+\!5r\!+\!
2s$ va\-rious variables, therefore its length is not more than $9m\cdot\lceil\log m\rceil$ for 
all $X$ such that $|X|\!=\!n$ and \ $m\!=\!m(n)\!\geqslant\!\lceil\log n\rceil\!=\!s\!
\geqslant\!r$. For $1\!\leqslant\!k\!\leqslant\!m$, the tuples 
$\widehat{Y}_{k},\widehat{a}_{k}$, and $\widehat{b}_{k}$ also have $2m+\!5r+\!2s$ variables; 
and each of these tuples has the variables of the identical type, but not of nine ones; hence, 
their maximal second index equals $\lceil\log(2m\!+\!5r\!+\!2s)\rceil$. 

Nevertheless, $|\exists\widehat{Y}_{k}|\!=\!|\forall\widehat{a}_{k}|\!
=\!|\forall\widehat{b}_{k}|\!=\!\mathcal{O}(m\cdot\lceil\log m\rceil)$ for all $k$ according to 
our suggestions about $X$ and $m$. Hence, 
\[
|\exists\,\widehat{Y}_{m}\forall\,\widehat{a}_{m}\forall\,\widehat{b}_{m}\ldots
\exists\,\widehat{Y}_1\forall\,\widehat{a}_{1}\forall\,\widehat{b}_{1}\bigl\{\bigwedge
\limits_{1\leqslant\,k\leqslant\,m}\Theta_{k}\bigr\}| \ = \ \mathcal{O}(m^2(n)\cdot\lceil\log m(n)
\rceil).
\] 
\end{proof}

\begin{cor}\label{lfrOm1} 
There exist constants $D_{1}, \ D_2$, and $D_w$ such that Inequalities~\ref{in1} hold  
for every sufficiently long $X$ and $m(n)\!\geqslant\!\lceil\log n \rceil$, i.e.,  
\begin{eqnarray*}
 m(n)\!<\!|\Omega^{(m)}(X,P)|\leqslant 
D_1\!\cdot\!n\!\!\cdot\!\lceil\log n\rceil\!\cdot\!\lceil\log\lceil\log n\rceil\rceil\!+\! 
D_2\cdot m(n)\cdot\lceil\log m(n)\rceil+ \\
+D_w\cdot[(|P|\cdot m(n)+m^2(n))\cdot\lceil\log m(n)\rceil  
\end{eqnarray*} 
\end{cor}

\begin{proof} 
It follows from this lemma, Lemma \ref{lchi0}, and the description of $\chi(0)$ 
in Subsection \ref{s6.2} because $|\forall\widehat{y}_0|\!<\!|\forall\widehat{y}_T|\!=\!
\mathcal{O}(m\cdot\lceil\log m\rceil$).
\end{proof}

Thus, we have proven that the formula $\Omega^{(m)}(X,P)$ of form \ref{6.1} possesses the 
properties (i) (see Proposition \ref{Pr5}) and (ii) from Theorem \ref{Main}, i.e., it is 
modeling and has a relatively short length. The simulating formula is described by the 
Equation \ref{6.1} in an explicit form; this description allows us to design a polynomial 
algorithm relative to $\max\{n, m(n)\}$ for its construction based on Lemma \ref{sp-con}.  

\section{The space complexity of languages from the class $\mathbf{P}$}\label{s7}

A simple application of the modeling theorem to languages of class $\mathbf{P}$ (as is done in 
the second proof of Theorem \ref{Imp-k} for languages from class $^{(k+1)}\mathbf{EXP}$ in 
Subsection \ref{s2.2}\,) yields the uninteresting fact that these languages can be recognized in 
an almost linear space, i.e., with memory \ $\mathcal{O}(n\lceil\log n\rceil^2\lceil\log\lceil
\log n\rceil\rceil)$; this follows from Corollary \ref{lfrOm1} for $m(n)\!=\!C\lceil\log n\rceil
$. We will try to strengthen this statement, although that does not give equality to 
classes $\mathbf{P}$ and $\mathbf{L}$. The equality of these classes it is hardly possible at all.

When $m(n)\!=\!\lceil\log n^d\rceil$, then the longest component of the modeling formula \break
$\Omega^{(m)}(X,P)$ is the subformula $\chi_1(X)$ corresponding to the unchanged entries on 
the first tape (see Subsection \ref{s6.2}\,) since it contains the subformula 
$\bigwedge\limits_{1\leqslant\,(\widehat{\xi})\leqslant\,n}\psi_{1}(\widehat{\xi}\!\rightarrow\!
\chi(\widehat{\xi}))$ describing the input string $X\!=\!\langle\chi_1,\ldots,\chi_n\rangle$ 
and having length \ $\mathcal{O}(n\cdot\lceil\log n\rceil\cdot\lceil\log\lceil\log n\rceil
\rceil)$. How this description can essentially be abridged is perfectly incomprehensible. The 
universal quantifier cannot be applied here, as the input string may contain various symbols. 

Should we not try to remove this description entirely because we have a record of the input 
string on the first tape?

\begin{thm} \label{PvsLog}
Every language of class $\mathbf{P}$ can be recognized in an almost 
logarithmical space; more precisely, if a language $\mathcal{L}$ belongs to $DTIME[n^d]$, 
then there is a 2-DTM $P(\mathcal{L})$ that recognizes $\mathcal{L}$ with memory bounded by  
$C\!\cdot\!d\,^2\!\cdot\!\lceil\log^2n\rceil\cdot\lceil\log\lceil\log n\rceil\rceil)$ for the 
appropriate constant $C$.
\end{thm} 

\begin{proof} 
Let us notice that the programs of all Turing machines depend only on the 
working alphabet, recognized language, and the number of their tapes, but not on an input 
string; though, the program's execution depends on the input, as a rule, essentially. 

Suppose the formula $\Omega^{(m)}(X,P)$ is written for the simulation of the actions of 2-DTM $P$ 
for some input $X$ during the first $\exp m(n)$ steps. In that case, we denote its subformula, 
which is obtained if the component "$\chi_1(X) \ \&$" is removed, as $\Delta^{(m)}(n,P)$, and will 
name it {\it conditionally modeling formula}. The reduced formula $\Delta^{(m)}(n,P)$ depends only 
on the program $P$, the function $m$, and the length $n$ 
of input according to the description of the simulating formula --- see Remark \ref{Ind}. Its 
length is $\mathcal{O}((|P|\cdot d\cdot\lceil\log n\rceil+d\,^{2}\cdot\lceil\log n)\rceil^2\cdot
\lceil\log\lceil\log n\rceil\rceil)$ for $m(n)\!=\!\mathcal{O}(\lceil\log n^d\rceil)$. We note 
that formula $\chi(0)$ is only shortened to \ $\chi_2(0) \ \& \ \pi_0((0)_2,\rhd,\rhd,(0)_2,(0)_2)
$ (see Subsection \ref{s6.2}), but is not removed at all.

We recall that the working alphabet of all 2-DTMs is $\mathcal{A}$, and $\mathcal{A}_1$ is the 
alphabet of the natural language described in Subsection \ref{s3.2}. Both these alphabets are 
finite. We regard that the programs of the machines are also written in the alphabet 
$\mathcal{B}\rightleftharpoons\mathcal{A}\cup\mathcal{A}_1$ with a natural identification of 
the implication and the arrow, which is involved in the record of the instructions. In 
Subsection \ref{s4.1}, applied a uniform length binary coding $c=c(\mathcal{A},P)$ of 
all symbols of the working alphabet $\mathcal{A}$ and simultaneously all internal states of 
the given machine $P$. Now, we extend this coding to $c_1=c_1(\mathcal{B},P)$ to 
include the whole alphabet $\mathcal{A}_1$. We will continue to write down the formulae 
$\Omega^{(m)}(X,P)$ and $\Delta^{(m)}(n,P)$, applying the alphabet $\mathcal{A}_1$ 
and the encoding $c_1$ as described in Section \ref{s4} for the encoding $c$.

However, one can do without such an extension if one immediately assumes that the alphabet 
$\mathcal{A}$ includes the whole alphabet $\mathcal{A}_1$ since only its finiteness is used in 
all constructions. On the other hand, we can proceed wholly to the codes of alphabetic symbols 
from $\mathcal{A}_1$ using the symbols of the alphabet $\mathcal{A}$; however, this will only 
overload the proof with irrelevant details.

\begin{lem} 
For every program $P$, there is 2-DTM $M(P)$ that writes $P$ when $M(P)$ gets 
any input; by that, the $\rhd$ symbol is written as $(c_1\rhd)$, i.e., as $r$-tuple consisting 
of $r\!-\!1$ zeros and one $1$ (see Subsection \ref{s4.1}). 
\end{lem}

\begin{proof}  
Executing its instructions, the machine $M(P)$ writes down all the symbols involved in the 
record of $P$ one by one in the second tape. The first head remains in the zeroth place during 
the entire work, so one may say that $M(P)$ is running in the one-tape mode.

Using an example, let us explain how $M(P)$ works. Let $q_7\rhd,B\rightarrow q_4R,C,L$ be some 
instruction in the program $P$. The sense of this instruction is as follows. If, in the 
seventh state, the machine $P$ sees the $\rhd$ symbol on the first tape and the $B$ symbol on 
the second at some instant, then $P$ moves the first head to the right, writes $C$ instead 
of $B$ on the second tape, moves the second head to the left, and enters the state $q_4$. 

This instruction must be written on the second tape of machine $M(P)$ in the following 
form (the comma belongs to the natural language's alphabet): 
$q7,(c_1\rhd),B\rightarrow q4,R,C,L$. 
The first six instructions of $M(P)$ for writing this instruction are (here we regard that 
$M(P)$ has entered the state $q_{101}$ after the executing of the current instruction; the 
ultimate non-empty symbol on the second tape is $D$, and the second head aims at this $D$):
\begin{eqnarray*}
&(1) \quad q_{101}\,\rhd,D\rightarrow q_{102}S,D,R \phantom{aaaa} & (2) \quad  q_{102}\,
\rhd,\Lambda\rightarrow q_{103}\,S,q,R \\
&(3) \quad q_{103}\,\rhd,\Lambda\rightarrow q_{104}S,7,R \phantom{aaaa} & (4) \quad    q_{104}
\,\rhd,\Lambda\rightarrow q_{105}S,,,R\\
&(5) \quad q_{105}\,\rhd,\Lambda\rightarrow q_{106}S,(,R \phantom{aaaa}  & (6) \quad   q_{106}
\,\rhd,\Lambda\rightarrow q_{107}S,0,R\\
\end{eqnarray*} 
The result of the actions of these instructions is the string $q7,(0$. 
\end{proof}

\begin{lem} 
For every program $P$ and each indexed-constructible function $m$, there is 
2-DTM $M_1(P)$ that writes the conditionally modeling formula $\Delta^{(m)}(n,P)$ when $M_1(P)$ 
gets any input $X$ having $n$ in length. For this purpose, $M_1(P)$ uses $\mathcal{O}(|
\Delta^{(m)}(n,P)|)$, i.e., $\mathcal{O}((|P|\cdot m(n)+m^2(n))\cdot\lceil\log m(n)\rceil)$ the 
memory cells for the sufficiently long $X$.
\end{lem}

\begin{proof} 
At first, the machine $M_1(P)$ activates its subprogram $M(P)$, which writes 
down the program $P$. Then, $M_1(P)$ calculates the value $m(n)$ and writes it in the binary 
representation on the second tape to the right of the string $P$. Next, it writes the formula 
$\Delta^{(m)}(n,P)$ using the strings $P$ and $(m(n))_2$. The machine $M(P)$ again runs in the 
one-tape mode in this and last stage. In the last stage,  $M_1(P)$ erases the strings $P$ and 
$(m(n))_2$ and shifts left the formula $\Delta^{(m)}(n,P)$. 
\end{proof}

Let us return to the proof of the theorem. Suppose some 2-DTM $P_0$ recognizes a 
language $\mathcal{L}$  within time $n^d$. One can regard the number $d$ as an integer. We 
will describe the machine $\mathcal{M(L)}$, which recognizes $\mathcal{L}$ in space 
$\mathcal{O}(\lceil\log^2 n)\rceil\cdot\lceil\log\lceil\log n\rceil\rceil)$. 

In Sections 4-6, one implicitly assumed the presence of 2-DTM, which writes the code of any 
symbol of the working alphabet on the second tape when the first head is aimed at this symbol. 
This was not stated due to the obviousness. It is not difficult to add 
several instructions to this program that check whether the clause of a kind \ 
$\psi_1(\widehat{\zeta}\rightarrow\alpha)=\widehat{X}\!\equiv\!\widehat{\zeta}\rightarrow 
\widehat{F}\!\equiv\!c\widehat{\alpha}$ is true, i.e., whether the $\alpha$ symbol is located 
in the $(\widehat{\zeta})$th cell of the first tape. The actions of this subroutine 
$M_2(\widehat{\zeta},\alpha)$ can be such: it shifts the head to the $(\widehat{\zeta})$th 
position on the first tape, reads some symbol $\beta$ there, writes the code $c\beta$ on the 
second tape, and then compares it with the $c\alpha$ code. Similarly, $M_2(\zeta_1,\alpha_1)$ 
can check the truth of the fragments of the timers $\pi_t((i)_2,\alpha_1,\alpha_2,\zeta_1,
\zeta_2)$ of any color $t$.

The last procedure we need is $M_3(\phi)$, which determines whether a given Boolean sentence 
$\phi$ is true. This machine uses a complete enumeration (or almost complete when using, for 
example, the branch and bound method) of possible values of the variables included in the 
record $\phi$. After generating the next set of variable values, it substitutes these values 
into $\phi$ and checks the truth of all subformulas from $\phi$. This algorithm requires 
exponential time but is feasible in a linearly bounded (relative to the length of $\phi$) 
space.

Thus, let the machine $\mathcal{M(L)}$ get an input $X$. In the first stage, it turns on 
subroutine $M_1(P_0)$, which writes the conditionally modeling formula $\Delta^{(m)}(n,P_0)$ for 
$m(n)=d\cdot\lceil\log n\rceil$. Let us notice that the formulae $\Delta^{(m)}(n, P_0)$ and 
$\Omega^{(m)}(X,P_0)$ contain the same variables. Grounding on this, the recognizing machine 
$\mathcal{M(L)}$ activates the modified procedure $M_3(\Omega^{(m)}(X,P_0))$ to check the 
validness of the sentence $\Omega^{(m)}(X, P_0)$. This modification lies in that the check of 
the truth of the timers $\pi_t((i)_2,\alpha,\alpha_2,\zeta,\zeta_2)$ and clauses 
$\psi_1(\widehat{\zeta}\rightarrow\alpha)$ is carried out by applying subroutine 
$M_2(\widehat{\zeta},\alpha)$. If the sentence $\Omega^{(m)}(X,P_0)$ occurs true, then 
this means that the machine $P_0$ accepts the input $X$ by Theorem \ref{Main}; hence, the string 
$X$ belongs to $\mathcal{L}$; therefore $\mathcal{M(L)}$ also accepts $X$. If not, $\mathcal{M(L)}
$ rejects $X$ since the word $X$ does not belong to the language $\mathcal{L}$. 
\end{proof}

\section{Results and Discussion}\label{s8}

\paragraph*{Motivation}
The basic motivation for writing this work was the author's desire to significantly improve 
the lower estimate of the computational complexity of equivalency-nontrivial decidable 
theories from \cite{Lat22}. For this purpose, the possibility of reducing the formula, which 
simulates lengthy computations, was analyzed. Three fragments were the longest for the modeling 
period $\exp n$, i.e., when $m(n)\!=\!n$. 

The first long component corresponds to the representation of the number $t\pm1$ in binary 
notation. If the binary representation $(t)_2$ of the number $t$ has length $n\lceil\log n
\rceil$, then in \cite{Lat22}, the binary representations of the numbers $t\pm1$ have length 
$\mathcal{O})(n\lceil\log n\rceil)^2)$ (Lemma 3.1 in Subsection 3.3). Introducing corrective 
additives, it was possible to significantly reduce the length of recording these numbers in 
Subsection \ref{s3.2} (compare Lemmata 3.1 in \cite{Lat22} and \ref{ll2} in this paper).

The second long fragment in \cite{Lat22}, the formula $\chi(0)$, describes the initial 
configuration, including the input string $X$. It was $\mathcal{O}(n\cdot n\lceil\log n\rceil)
$ symbols in length. Now, it has the length $\mathcal{O}(n\cdot\lceil\log n\rceil\cdot\lceil\log
\lceil\log n\rceil\rceil)$ through the usage of an input tape on which nothing can be written or 
erased. This is mainly why the author began to model the actions of two-tape machines in this 
paper. He has only seen the possibility of proving Theorem \ref{PvsLog} in the paper 
writing-process. 

Nevertheless, the other long component, the $\Upsilon(m)$ formula (see Lemma \ref{lfrOm}), remains 
here of the same length. Apparently, it is possible to significantly reduce the size of the 
description of the exponential (relative to $m(n)$) number of steps only by abandoning the 
usage of the elegant construction of Meyer and Stockmeyer, but that has not yet been possible 
to do accurately enough.

\paragraph*{Relativization and nondeterministic computations} One running step of a
non-determi\-nis\-tic machine can not only be described by somewhat complicating the formula 
$\varphi(k)$ from Subsection \ref{s4.2} but even simulated, i.e., it is possible to prove 
analogs of Propositions \ref{Pr2}(ii) and \ref{Pr3} for this modified formula. 
However, it is  very challenging to adequately model long, non-deterministic computations through 
short formulae. 

The point is that the method of modeling the actions of Turing machines described in this 
paper and \cite{Lat22} is \emph{point-wise} (or \emph{strictly local}), in contrast to Cook's 
formula from \cite{Coo71}. For example, when proving Proposition \ref{Pr2}, we were able to 
explicitly "see" only three cells on the second tape, i.e., set appropriate values for the 
variables included in the description of these squares; namely, these are cells with numbers 
$\zeta_2$, $\zeta_2(\beta_2)$, and $\mu$. At the same time, the proof of an analog of 
Proposition \ref{Pr4} may require simultaneous "consideration" of almost the entire zone with a 
length of $\exp(m(n))$ cells for non-deterministic computations. The description of the 
corresponding example takes several pages and awaits a separate publication.

Relativization of the main theorem on modeling (Theorem \ref{Main}) seems to the author, firstly 
all, to be meaningless. Secondly, it is impossible. Indeed, in an exponential number of steps, the 
machine can record on the oracle's tape a string whose length is not bounded by the value of the 
polynomial. Certainly, the recording process can be described by formulae similar to $\Phi^{(0)}
(\widehat{y}_t,\widehat{y}_{t+1})$ (see Subsections \ref{s4.2} and \ref{s4.3}\,). But to describe 
the query to the oracle, it will be necessary to use a formula of exponential length since the 
query chain may be heterogeneous in structure, and as a result, the universal quantifier is not 
applicable here. In addition, this query string will only be known to us at the instant of the 
query.

\paragraph*{One can prove another way}
Probably the equality of the classes $\mathbf{PSPACE}$ and $\mathbf{EXP}$ can proven in 
another way. Instead of the $QSAT$ problem, one can theoretically take any other 
$\mathbf{PSPACE}$-complete problem and prove its completeness for the $\mathbf{EXP}$ class. 
However, such a construction's practical realizability should also be considered.
The desired problem must have a fairly simple description; at the same time, the formal 
language of this problem must be expressive enough so that it can be used to model the 
exponential computations of Turing machines. From this point of view, the word problem for 
suitable automaton groups or automaton semi-groups is promising. For these problems, their 
completeness in the $\mathbf{PSPACE}$ class has recently been proven \cite{An-Ro-Wa, Wa-We}. 
Of course, choosing a $\mathbf{PSPACE}$-complete problem for this purpose is a very subjective 
thing.

\bibliographystyle{alphaurl}
\bibliography{correspondence}
\end{document}